\documentclass[sigconf]{sig-alternate-per}

\usepackage{lineno, amsmath, bbm, graphicx, caption, subcaption, float, mathtools, xcolor, comment}

\DeclarePairedDelimiter{\floor}{\lfloor}{\rfloor}
\usepackage[font=small,skip=0pt]{caption}

\newtheorem{theorem}{Theorem}
\newtheorem{corollary}{Corollary}

\makeatletter
\g@addto@macro\normalsize{%
  \setlength\abovedisplayskip{2pt}
  \setlength\belowdisplayskip{2pt}
}
\makeatother

\def\sharedaffiliation{%
\end{tabular}
\begin{tabular}{c}}
\begin{document}
\conferenceinfo{IFIP WG 7.3 Performance 2017.}{Nov.\ 14-16, 2017, New York, NY USA}
\title{Straggler Mitigation by Delayed Relaunch of Tasks}

\numberofauthors{3}
\author{
  Mehmet Fatih Akta\c{s} \\
  mehmet.aktas@rutgers.edu
  \alignauthor
  Pei Peng \\
  pei.peng@rutgers.edu
  \alignauthor
  Emina Soljanin \\
  emina.soljanin@rutgers.edu
  \sharedaffiliation
    \affaddr{Department of Electrical and Computer Engineering, Rutgers University}
}

\maketitle

\abstract{
  Redundancy for straggler mitigation, originally in data download and more recently in distributed computing context, has been shown to be effective both in theory and practice. Analysis of systems with redundancy has drawn significant attention and numerous papers have studied pain and gain of redundancy under various service models and assumptions on the straggler characteristics. We here present a cost (pain) vs. latency (gain) analysis of using simple replication or erasure coding for straggler mitigation in executing jobs with many tasks. We quantify the effect of the tail of task execution times and discuss tail heaviness as a decisive parameter for the cost and latency of using redundancy. Specifically, we find that coded redundancy achieves better cost vs. latency tradeoff than simple replication and can yield reduction in both cost and latency under less heavy tailed execution times. We show that delaying redundancy is not effective in reducing cost and that delayed relaunch of stragglers can yield significant reduction in cost and latency. We validate these observations by comparing with the simulations that use empirical distributions extracted from Google cluster data.
}

\section{Introduction and Model}
\noindent\textbf{Motivation:}
Distributed (computing) systems aim to attain scalability through parallel execution of multiple tasks constituting a job. Each task is run on a separate node, and the job is completed only when the slowest task is finished. It has been observed that task execution times have significant variability, e.g., because of multiple job resource sharing, power management \cite{dean2013tail}. The slowest tasks that determine the job execution time are known as "stragglers''.

Two common performance metrics for distributed job execution are 1) \emph{Latency,} measuring the execution time, and 2) \emph{Cost,} measuring the resource usage. Job execution is desired to be fast and with low cost, but these are conflicting objectives. Replicating tasks and running the replicas over separate nodes has been shown to be effective in mitigating the effect of stragglers on latency \cite{ananthanarayanan2013effective}, and is used in practice \cite{dean2008mapreduce}. Recent research proposes to delay replication in order to reduce the cost \cite{wang2015using}, and clone only the tasks that at some point appear to be straggling.

Erasure coding is a more general form of redundancy than simple replication, and it has been considered for straggler mitigation in both data download \cite{joshi2015queues} and, more recently, in distributed computing context \cite{dutta2016short, CodedGradientDescent:HalbawiAS17}.
We took this line of work further by analyzing the effect of coding and replication on the tradeoff between cost and latency as in \cite{ourmamapaper}. We examined whether the redundancy should be simple replication or coding, and when it should be introduced. We here extend the cost and latency analysis in \cite{ourmamapaper} for systems that use redundancy together with task relaunch.
\\[1ex]
\noindent\textbf{System Model:}
In our system, a job is split into $k$ tasks. Job execution starts with launching all its $k$ tasks, and the redundancy is introduced only if the job is not completed by some time $\Delta$. Note that we don't consider queueing of jobs or tasks; all tasks start service together.

In replicated-redundancy $(k, c, \Delta)$-system, if the job still runs at time $\Delta$, then $c$ replicas for each remaining task are launched. In coded-redundancy $(k, n, \Delta)$-system, if the job still runs at time $\Delta$, $n-k$ redundant parity tasks are launched where completion of any $k$ of all launched tasks results in total job completion (see Fig.~\ref{fig:fig_delayed_red}). Note that this assumption does not impose severe restrictions. Any linear computing algorithm can be implemented with this $k$-out-of-$n$ structure simply by using linear erasure codes. Particular examples can be found in e.g., \cite{dutta2016short, CodedGradientDescent:HalbawiAS17} and references therein. If system implements task relaunch, then tasks that remain running at time $\Delta$ are canceled, and fresh  copies are launched in their place immediately together with the redundant tasks.
\begin{figure}[t]
  \centering
  \includegraphics[width=0.5\textwidth, keepaspectratio=true]{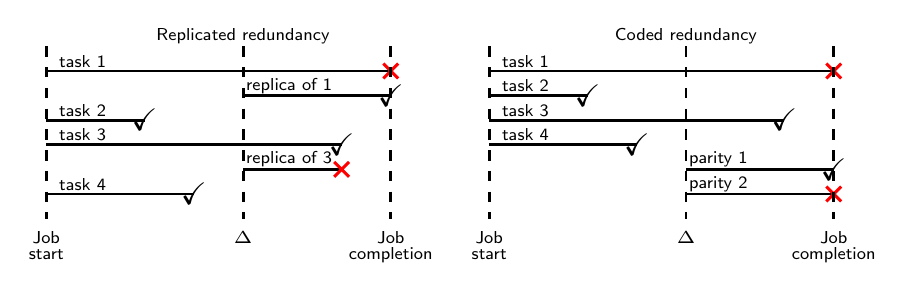}
  \caption{A job with four tasks is executed with delayed redundancy. Check marks represents task completion and crosses represents task cancellation. (With replication,  exact clones of the straggler tasks are introduced, while with coding parity tasks can be used as a ``clone'' for any task, therefore, stragglers do not have to be tracked down.)}
  \label{fig:fig_delayed_red}
\end{figure}

\begin{figure*}[t]
  \centering
  \begin{subfigure}[]{.32\textwidth}
    \centering
    \includegraphics[width=1\textwidth, keepaspectratio=true]{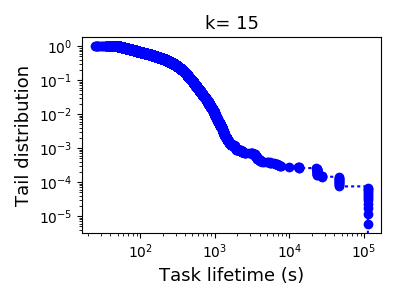}
  \end{subfigure}
  \begin{subfigure}[]{.32\textwidth}
    \centering
    \includegraphics[width=1\textwidth, keepaspectratio=true]{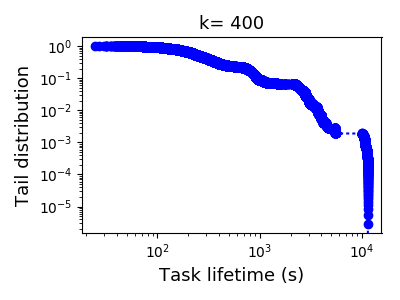}
  \end{subfigure}
  \begin{subfigure}[]{.32\textwidth}
    \centering
    \includegraphics[width=1\textwidth, keepaspectratio=true]{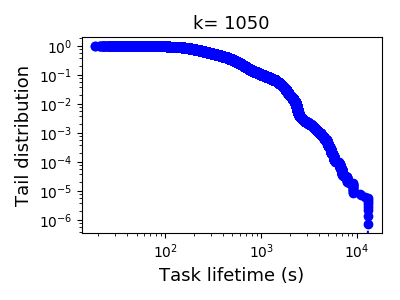}
  \end{subfigure}
  \caption{Empirical tail distribution of task completion times for Google cluster jobs with number of tasks $k=15, 400, 1050$.}
  \label{fig:plot_google_empiricaltail}
\end{figure*}

We assume that task execution times are iid and follow one of the two canonical distributions: 1)
shifted exponential  ${\it SExp}(D, \mu)$ modeling tasks that take some positive minimum time $D$ and have an exponential tail with decaying rate $\mu$ modeling the randomness inherent in the system and 2) $\textit{Pareto}(\lambda, \alpha)$ with positive minimum value $\lambda$ and  a power law tail index $\alpha$. Pareto is a canonical heavy tailed distribution that is observed to fit task execution times in real computing systems \cite{dean2013tail, reiss2012towards}.

Fig.~\ref{fig:plot_google_empiricaltail} plots the empirical tail distribution of task completion times\footnote{Task lifetimes are calculated as the difference between timestamps for SCHEDULE and FINISH events in \cite{reiss2011google}.} for jobs with $15$, $400$, and $1050$ tasks in the Google Trace data \cite{reiss2011google}. Note that the $x$ and $y$ axes are in log scale, and thus an exponential tail would have appeared as a curve decaying exponentially while a true power law tail (e.g., Pareto) would have pronounced a linear decay at a constant rate. Empirical tail distributions in the figure exhibit exponential decay at small values and a trend similar to linear decay at larger values. Note that the steep decay at the far right edge is due to bounded support of the distributions. Even though we cannot conclude that these empirical distributions are distributed as Pareto, they clearly exhibit more variability and have heavier tail than an Exponential.

We define the cost of job execution as the sum of the lifetimes of all tasks (including redundant ones) involved in the job execution. This definition reflects ``pay for resources'' pricing, which is the most commonly used model in cloud services offered by Amazon, Google App Engine, Windows Azure \cite{al2013cloud}. There are two main setups that define cost: 1) Cost with task cancellation $C^c$; remaining outstanding tasks are canceled upon the job completion, which is a viable option for distributed computing with redundancy, 2) Cost without task cancellation $C$; tasks remaining after job completion run until they complete, which for instance is the only option for data transmission over multi-path network with redundancy.

In this paper, we analyze the effect of replicated and coded redundancy on cost and latency tradeoff. Specifically, we present exact expressions for expected latency and cost for redundancy with and without task relaunch. Using these expressions, we show the correlation of pain and gain of redundancy with the tail heaviness of the task execution times.
\\[1ex]
\textbf{Summary of Observations:} Coding allows us to increase degree of redundancy with finer steps than replication, which translates into greater achievable cost vs.\ latency region. Delaying redundancy is not effective in trading off latency for cost. Therefore, primarily the degree of redundancy should be tuned for the desired cost and latency values. Coding is shown to outperform replication in terms of cost and latency together. When the task execution times are heavy tailed, redundancy can reduce cost and latency simultaneously, where the reduction depends on the tail heaviness. For heavy tailed tasks, we show that relaunching tasks, even without any redundancy, is sufficient to return significant reduction in cost and latency.

\vspace{1ex}
\noindent\textbf{Notation:}
$T$ and $C$ denote latency and cost of job execution.
$H_n$ is the $n$th harmonic number defined as $\sum_{i=1}^n \frac{1}{i}$ for $n \in Z^+$ or equivalently as $\int_0^1 \frac{1-x^n}{1-x} dx$ for $n \in R$. $H_{n^2}$ denotes the generalized harmonic number of order $n$ of two defined as $\sum_{i=1}^n \frac{1}{i^2}$. Incomplete Beta function $B(q;m,n)$ is defined for $q \in [0,1]$, $m, n \in R^+$ as $\int_0^q u^{m-1}(1-u)^{n-1} du$, and Beta function as $B(m,n) = B(1;m,n)$. Gamma function $\Gamma(x)$ is defined as $\int_0^{\infty} u^{x-1}e^{-u}du$ for $x \in R$ and as $(x-1)!$ for $x \in Z^+$.

\section{Latency and Cost Analysis}
In a previous work \cite{ourmamapaper}, we concluded that delaying replicated or coded redundancy is not effective to reduce cost. This conclusion is based on the observation that delaying redundancy can bring reduction in cost (gain) only after significant increase (pain) in latency, at which point one can achieve less latency for the same cost by simply reducing the level of redundancy. This section gives cost vs.\ latency analysis of zero-delay redundancy systems, where cost and latency are expressed in terms of level of redundancy $c$ or $n$ and parameters of task execution time distribution.

\begin{figure*}[t]
  \centering
  \begin{subfigure}[]{.32\textwidth}
    \centering
    \includegraphics[width=1\textwidth, keepaspectratio=true]{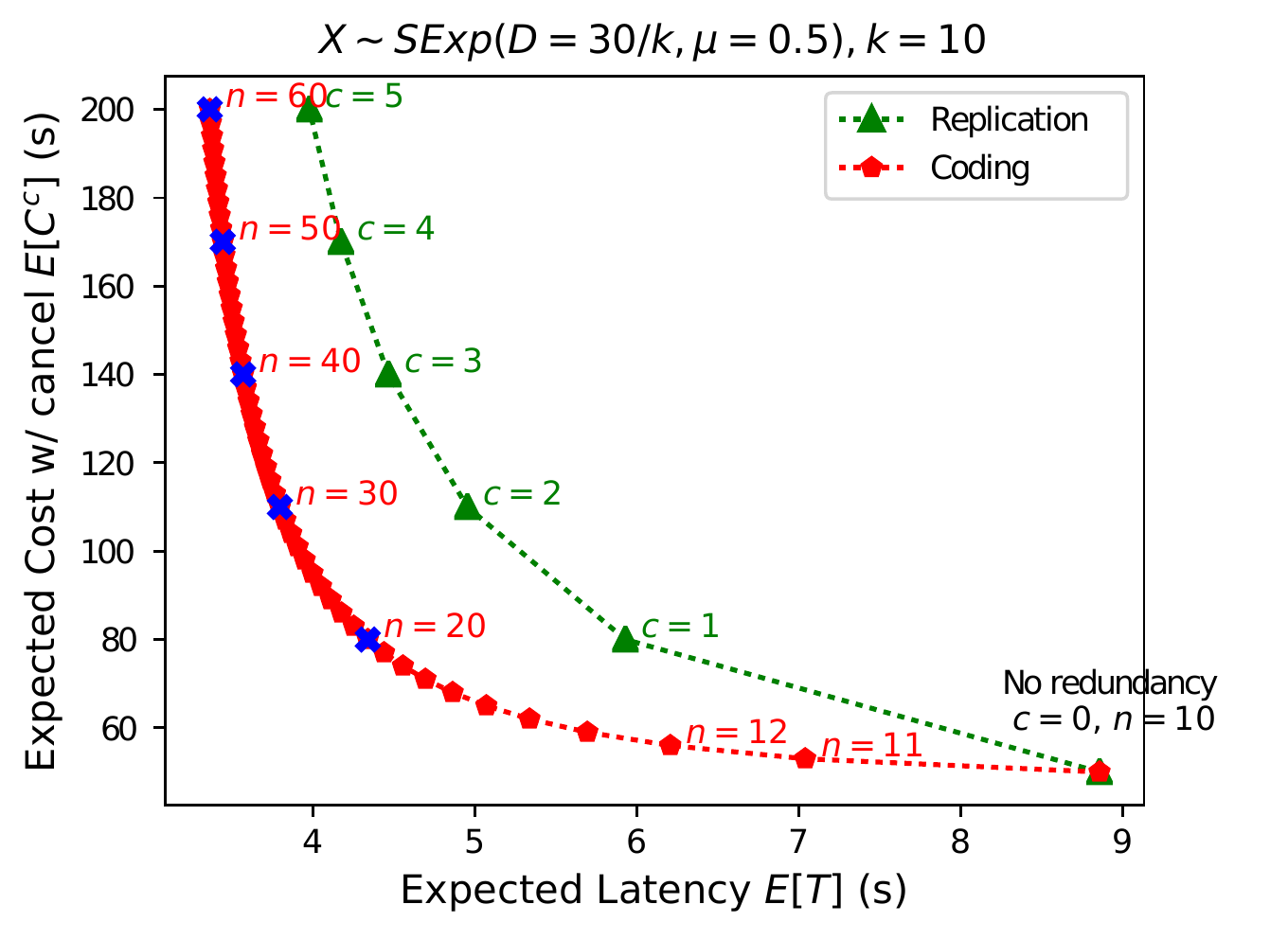}
  \end{subfigure}
  \begin{subfigure}[]{.32\textwidth}
    \centering
    \includegraphics[width=1\textwidth, keepaspectratio=true]{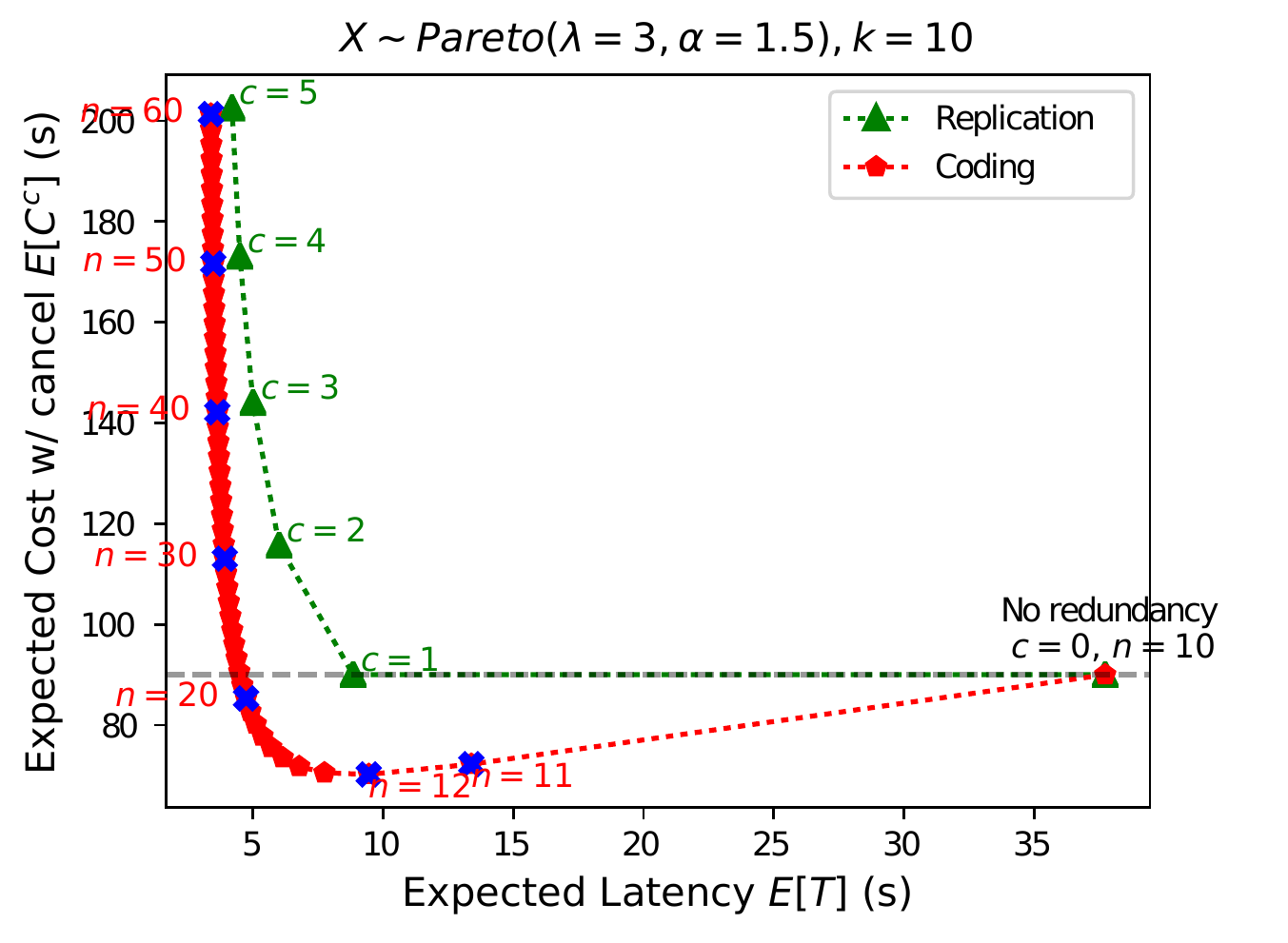}
  \end{subfigure}
  \begin{subfigure}[]{.32\textwidth}
    \centering
    \includegraphics[width=1\textwidth, keepaspectratio=true]{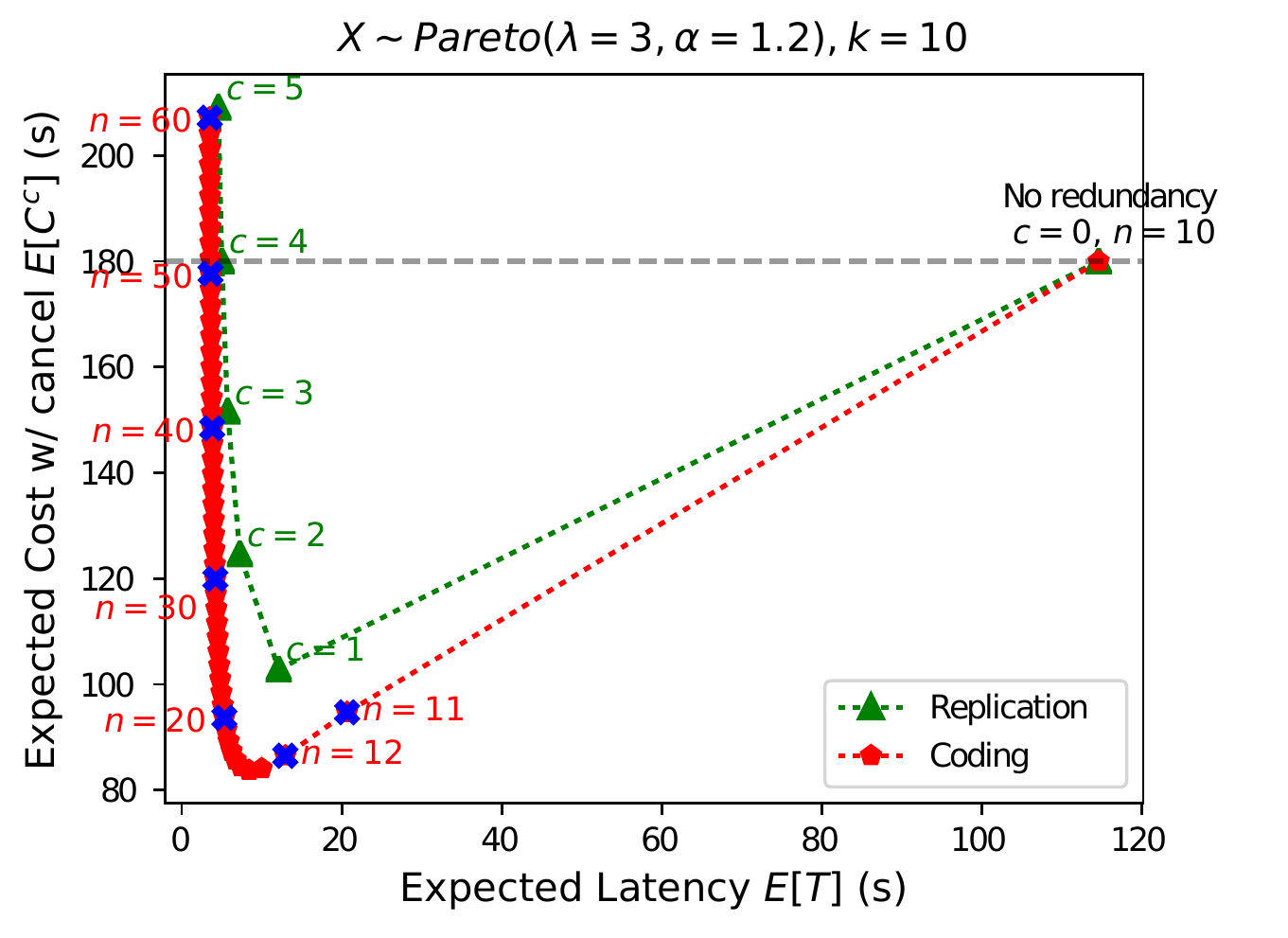}
  \end{subfigure}
  \caption{Expected cost vs.\ latency for zero-delay redundancy where redundancy levels $c$ and $n$ vary along the curves. Tail heaviness increases from left to right. The heavier the tail is, the higher the maximum reduction in expected cost and latency is.}
  \label{fig:plot_zerodelay_reped_vs_coded_k_10}
\end{figure*}

Thm.~\ref{thm_zerodelay_red_E_T__E_C} gives exact expressions for the expected cost and latency under zero-delay redundancy, and Fig.~\ref{fig:plot_zerodelay_reped_vs_coded_k_10} shows comparison between replication and coding for varying level of redundancy. Under both $SExp$ and $Pareto$ task execution times, coding always achieves better expected cost and latency than replication.
\begin{theorem}
  Let expected latency and cost with task cancellation be $T_{(k,c)}$, $C_{(k,c)}$ for zero-delay replicated redundancy, and $T_{(k,n)}$, $C_{(k,n)}$ for zero-delay coded redundancy. Under task execution time $X \sim SExp(\frac{D}{k}, \mu)$, we have
  \begin{equation*}
  \begin{split}
    E[T_{(k,c)}] &= \frac{D}{k} + \frac{H_k}{(c+1)\mu}, \quad\quad E[C_{(k,c)}] = (c+1)D + \frac{k}{\mu}, \\
    E[T_{(k,n)}] &= \frac{D}{k} + \frac{1}{\mu}(H_n-H_{n-k}), \quad E[C_{(k,n)}] = \frac{nD}{k} + \frac{k}{\mu}.
  \end{split}
  \label{eq:eq_zerodelay_reped_SExp_E_T__E_C}
  \end{equation*}
  
  Under task execution time $X \sim Pareto(\lambda, \alpha)$, we have
  \begin{equation*}
  \begin{split}
    E[T_{(k,c)}] &= \lambda k!\frac{\Gamma(1-((c+1)\alpha)^{-1})}{\Gamma(k+1-((c+1)\alpha)^{-1})}, \\
    E[C_{(k,c)}] &= \lambda k(c+1)\frac{(c+1)\alpha}{(c+1)\alpha-1}, \\
    E[T_{(k,n)}] &= \lambda\frac{n!}{(n-k)!}\frac{\Gamma(n-k+1-\alpha^{-1})}{\Gamma(n+1-\alpha^{-1})}, \\
    E[C_{(k,n)}] &= \lambda\frac{n}{\alpha-1}\Bigl[\alpha - \frac{\Gamma(n)}{\Gamma(n-k)}\frac{\Gamma(n-k+1-\alpha^{-1})}{\Gamma(n+1-\alpha^{-1})}
    \Bigr].
  \end{split}
  \label{eq:eq_zerodelay_reped_Pareto_E_T__E_C}
  \end{equation*}
  \label{thm_zerodelay_red_E_T__E_C}
\end{theorem}
\begin{proof}[Sketch]
  For replicated redundancy, lifetime of each task is $X_{c+1:1} \sim Pareto(\lambda, (c+1)\alpha)$. Then the expected values of latency $(X_{c+1:1})_{k:k}$ and cost $\sum_{i=1}^k (X_{c+1:1})_{k:i}$ follow from first principles of order statistics.
  Same calculations apply for expectation of latency $X_{n:k}$ and cost $\sum_{i=1}^k X_{n:i} + (n-k)X_{n:k}$ of execution with coded redundancy.
\end{proof}

Under exponential tail, adding redundancy reduces latency but increases cost. In \cite{wang2015using}, replicated redundancy is demonstrated to reduce both cost and latency under heavy tailed task execution time. Fig.~\ref{fig:plot_zerodelay_reped_vs_coded_k_10} plots latency and cost reduction for exponential and heavy tailed task execution times using the expressions given in Thm.~\ref{thm_zerodelay_red_E_T__E_C}. It illustrates the intuitive conclusion that redundancy can yield greater reduction in cost and latency under heavier tailed task execution times.

It is worth to discuss how and why cost reduction matters. Cost, as is defined here, reflects the amount of resource time used to execute a job. Reduction in cost then means running the same job by occupying less area in the space of overall system capacity, which then allows fitting more jobs per area, and hence higher system throughput. To illustrate with a simple example, consider a First-Come First-Served queue with single server under heavy traffic, that is, server never goes idle since there always exists a job to be served. Then, average throughput is the reciprocal of average job service time. In this case, cost of serving a job is simply its service time and system throughput increases with reduced average cost. In the case of systems with many servers that serve jobs with multiple tasks, same principles apply and cost reduction opens up space to execute more jobs per time unit.

Although exact expressions are formidable to derive, second moments of latency and cost can be computed as described in Thm.~\ref{thm_zerodelay_red_E_T_2__E_C_2}. Second moments enable us to compute the standard deviation of latency and cost. Fig.~\ref{fig:plot_zerodelay_reped_vs_coded_wstdev_k_15} plots expected cost and latency with error bars of width equal to the standard deviation in respective dimensions. Standard deviation, hence the variability of cost and latency decreases with increasing level of redundancy as expected. Reduction in variability for the same level of redundancy is greater with coding compared to using replication.

\begin{theorem}
  For $X \sim Exp(\mu)$ and $j \geq i$
  \begin{equation*}
    E[X_{n:i}X_{n:j}] = \frac{1}{\mu^2}[H_{n^2} - H_{(n-i)^2} + (H_n - H_{n-i})(H_n - H_{n-j})].
  \label{eq:eq_Exp_orderstat_joint_moment}
  \end{equation*}
  Then, under task execution time $X \sim SExp(D/k, \mu)$, second moments of latency and cost for zero-delay replicated and coded redundancy systems can be computed as
  \begin{equation*}
  \begin{split}
	E[T_{(k,c)}^2] &= \left(\frac{D}{k} + \frac{H_k}{(c+1)\mu}\right)^2 + \frac{H_{k^2}}{(c+1)^2\mu^2}, \\
    E[C_{(k,c)}^2] &= ((c+1)D)^2 + 2D(c+1)\frac{k}{\mu} \\
    &\quad + (c+1)^2 \sum_{i,j=1}^k E[Y_{n:i}Y_{n:j}] \\
    E[T_{(k,n)}^2] &= \frac{H_{n^2} - H_{(n-k)^2}}{\mu^2} + \left(\frac{D}{k} + \frac{H_n - H_{n-k}}{\mu}\right)^2, \\
    E[C_{(k,n)}^2] &= \left(n\frac{D}{k}\right)^2 + 2n\frac{D}{\mu} + (n-k)^2 E[X_{n:k}^2] \\
    &\quad + 2(n-k)\sum_{i=1}^k E[X_{n:i}X_{n:k}] + \sum_{i,j=1}^k E[X_{n:i}X_{n:j}].
  \end{split}
  \label{eq:eq_zerodelay_reped_red_E_T_2__E_C_2}
  \end{equation*}
  where $Y \sim Exp((c+1)\mu)$.
  
  \vspace{2ex}
  For $X \sim Pareto(\lambda, \alpha)$, given $\alpha > \max\{2(n-i+1)^{-1}, (n-j+1)^{-1}\}$ and $j \geq i$
  \begin{equation*}
  \begin{split}
    E[X_{n:i}X_{n:j}] = &\lambda^2\frac{n!}{\Gamma(n+1-2/\alpha)} \\
    &\times \frac{\Gamma(n-i+1-2/\alpha)}{\Gamma(n-i+1-1/\alpha)} \frac{\Gamma(n-j+1-2/\alpha)}{\Gamma(n-j+1)}.
  \end{split}
  \label{eq:eq_Pareto_orderstat_joint_moment}
  \end{equation*}
  Then, under task execution time $X \sim Pareto(\lambda, \alpha)$, second moments of latency and cost for zero-delay replicated and coded redundancy systems can be computed as
  \begin{equation*}
  \begin{split}
	E[T_{(k,c)}^2] &= E[Y_{k:k}^2], \\
    E[C_{(k,c)}^2] &= (c+1)^2 \sum_{i,j=1}^k E[Y_{k:i}Y_{k:j}], \\
    E[T_{(k,n)}^2] &= E[X_{n:k}^2] \\
    E[C_{(k,n)}^2] &= (n-k)^2 E[X_{n:k}^2] + 2(n-k)\sum_{i=1}^k E[X_{n:i}X_{n:k}] \\
    &\quad + \sum_{i,j=1}^k E[X_{n:i}X_{n:j}].
  \end{split}
  \label{eq:eq_zerodelay_red_Pareto_E_T_2__E_C_2}
  \end{equation*}
  where $Y \sim Pareto(\lambda, (c+1)\alpha)$.
  \label{thm_zerodelay_red_E_T_2__E_C_2}
\end{theorem}
\begin{proof}[Sketch]
  Derivations follow from latency and cost expressions given in the proof of Thm.~\ref{thm_zerodelay_red_E_T__E_C}. Expressions of $E[X_{n:i}X_{n:j}]$ can be found for $X \sim {\it Exp}$ in Pg. 73 of \cite{OrderStat:Arnold08}, and for $X \sim {\it Pareto}$ in Pg. 62 of \cite{Pareto:Arnold15}.
\end{proof}

Under heavy tail, it is possible to reduce latency by adding redundancy and still pay for the baseline cost of running with no redundancy. Corollary \ref{cor_zerodelay_red_pareto_reduc_in_E_T_for_baseline_E_C} gives expressions for the minimum achievable expected latency $E[T_{\min}]$ without exceeding the baseline cost of running with no redundancy.

\begin{figure}[t]
  \centering
  \includegraphics[width=0.36\textwidth, keepaspectratio=true]{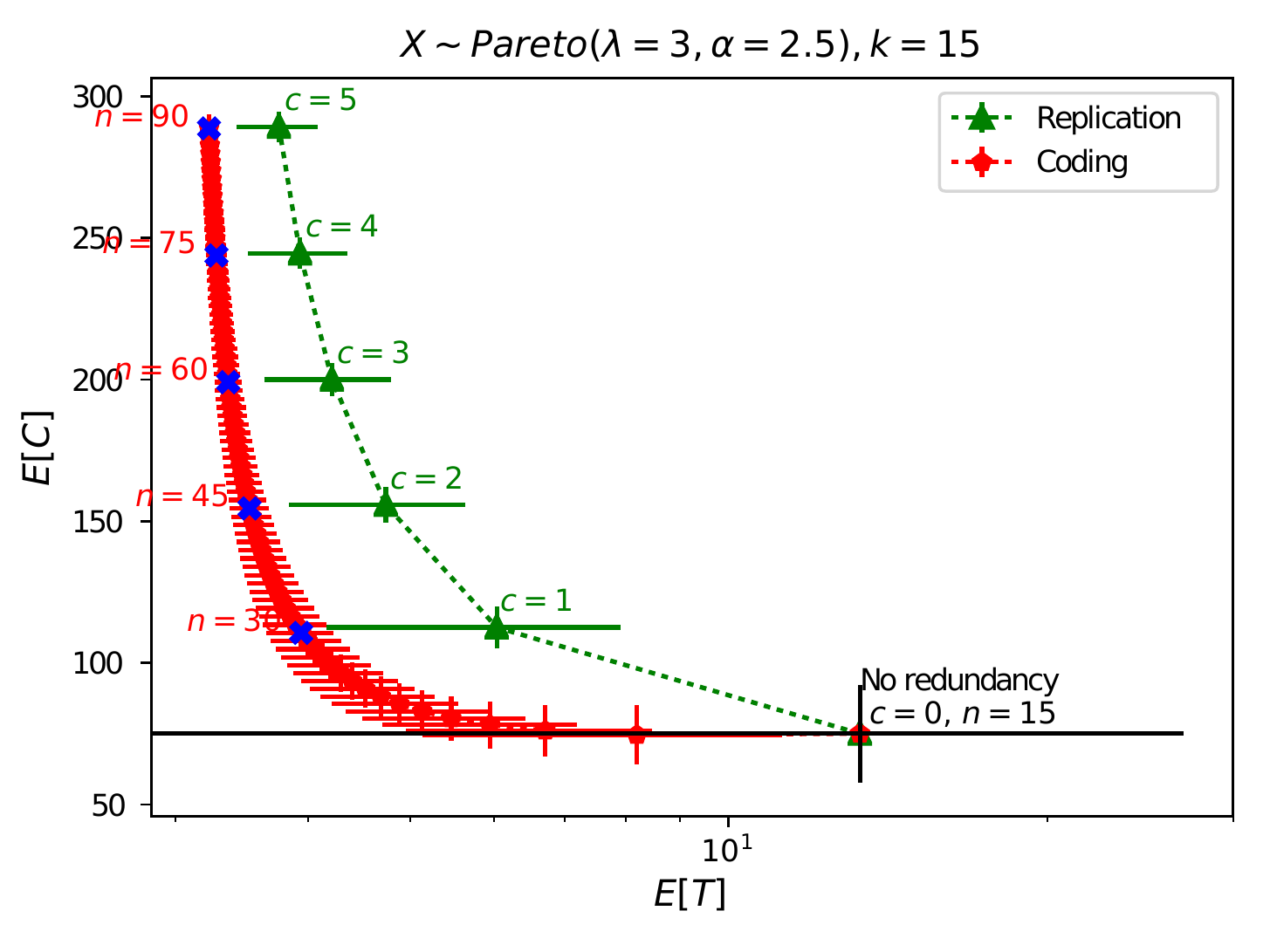}
  \caption{Expected cost vs.\ latency for zero-delay redundancy systems. Widths of horizontal error bars are equal to standard deviation of latency and widths of vertical bars are equal to standard deviation of cost.}
  \label{fig:plot_zerodelay_reped_vs_coded_wstdev_k_15}
\end{figure}

\begin{corollary}
  Under task execution time $X \sim \textit{Pareto}(\lambda, \alpha)$ in zero-delay replicated redundancy system, minimum latency $E[T_{\min}]$ that can be achieved without exceeding the baseline cost is,
  \begin{equation}
    E[T_{\min}] = \lambda k!\frac{\Gamma(1-(\alpha(c_{\max}+1))^{-1})}{\Gamma(k+1-(\alpha(c_{\max}+1))^{-1})}.
  \label{eq:eq_zerodelay_reped_pareto_reduc_in_E_T_for_base_E_C}
  \end{equation}
  where $c_{\max} = \max\{\floor*{\frac{1}{\alpha-1}}-1, 0\}$ and any reduction in latency without exceeding the baseline cost is possible only if $\alpha < 1.5$.
  For coded redundancy system,
  \begin{equation}
  \begin{split}
    E[T_{\min}] = E[T(n_{\max})].
  \end{split}
  \label{eq:eq_zerodelay_coded_pareto_reduc_in_E_T_for_base_E_C}
  \end{equation}
  where
  \begin{equation}
  \begin{split}
    n_{\max} &= \max\{n | E[T(n_{\max})] - \frac{E[T(k)]}{(n_{\max}-k)} - \alpha \leq 0 \}, \\
    E[T(n)] &= \lambda\frac{n!}{(n-k)!}\frac{\Gamma(n-k+1-\alpha^{-1})}{\Gamma(n+1-\alpha^{-1})}.
  \end{split}
  \end{equation}
  or an upper bound on $E[T_{\min}]$ is
  \begin{equation}
    E[T_{\min}] < \lambda\alpha + \lambda k!\frac{\Gamma(1-\alpha^{-1})}{\Gamma(k+1-\alpha^{-1})}.
  \label{eq:eq_zerodelay_coded_pareto_reduc_in_E_T_for_base_E_C__ineq}
  \end{equation}
\label{cor_zerodelay_red_pareto_reduc_in_E_T_for_baseline_E_C}
\end{corollary}
\begin{proof}
  First consider replicated redundancy $(k,c)$-system. Latency is a decreasing function of $c$, while cost may decrease up to $c_{\max}$ beyond which it increases with $c$. We would like to find $c_{\max}$ such that $E[C_{c=c_{\max}}] < E[C_{c=0}] < E[C_{c=c_{\max}+1}]$, then $E[T_{\min}]$ is simply $E[T_{c=c_{\max}}]$. We next obtain the range of $c$ for which $E[C_c] < E[C_{c=0}]$.
  \begin{equation*}
  \begin{split}
    E[C_c] < E[C_{c=0}] = \lambda k(\frac{(c+1)^2\alpha}{(c+1)\alpha-1} - \frac{\alpha}{\alpha-1}) < 0.
  \end{split}
  \end{equation*}
  which holds only when $c < \frac{1}{\alpha-1}-1$, and since $c$ is a non-negative integer, $c_{\max} = \max\{\floor*{\frac{1}{\alpha-1}}-1, 0\}$, from which  \eqref{eq:eq_zerodelay_reped_pareto_reduc_in_E_T_for_base_E_C} follows.
  
  Next consider a coded redundancy $(k,n)$-system. Expressions for the expected cost in this case do not allow to find $n_{\max}$ such that $E[C_{n=n_{\max}}] < E[C_{n=k}] < E[C_{n=n_{\max}+1}]$. Instead we can express $E[C]$ as a function of $E[T]$ as $E[C] = \lambda\frac{n}{\alpha-1}(\alpha - \frac{n-k}{n\lambda}E[T])$. We can then find an approximation for $E[T_{\min}]$ directly by relating $E[C_{n_{\max}}]$ to $E[C_{n=k}]$.\\ $E[C_{n_{\max}}] < E[C_{n=k}]$ holds only if $E[T_{n_{\max}}] < \lambda\alpha + \frac{E[T_{n=k}]}{n_{\max}-k}$, from which \eqref{eq:eq_zerodelay_coded_pareto_reduc_in_E_T_for_base_E_C} follows. Upper bound \eqref{eq:eq_zerodelay_coded_pareto_reduc_in_E_T_for_base_E_C__ineq} follows from this by setting $n_{\max} = k+1$.
\end{proof}

Fig.~\ref{fig:plot_zerodelay_red_pareto_reduc_in_E_T_for_baseline_E_C} illustrates that the maximum percentage reduction in latency (i.e., $\frac{E[T_0]-E[T_{\min}]}{E[T_0]}$; $E[T_{\min}]$ is latency without exceeding the baseline cost, $E[T_0]$ is the latency with no redundancy) depends on the tail of task execution time. As stated in Corollary \ref{cor_zerodelay_red_pareto_reduc_in_E_T_for_baseline_E_C}, this reduction is possible under replicated redundancy only when the tail index is less than $1.5$, in other words when the tail is quite heavy, while coding relaxes this constraint. Moreover, the threshold on $\alpha$ under replication is independent of the number of tasks $k$, while threshold increases with $k$ under coding, meaning that jobs with larger number of tasks can get reduction in latency at no cost even for lighter tailed task execution times.
\begin{figure}[h]
  \centering
  \includegraphics[width=0.35\textwidth, keepaspectratio=true]{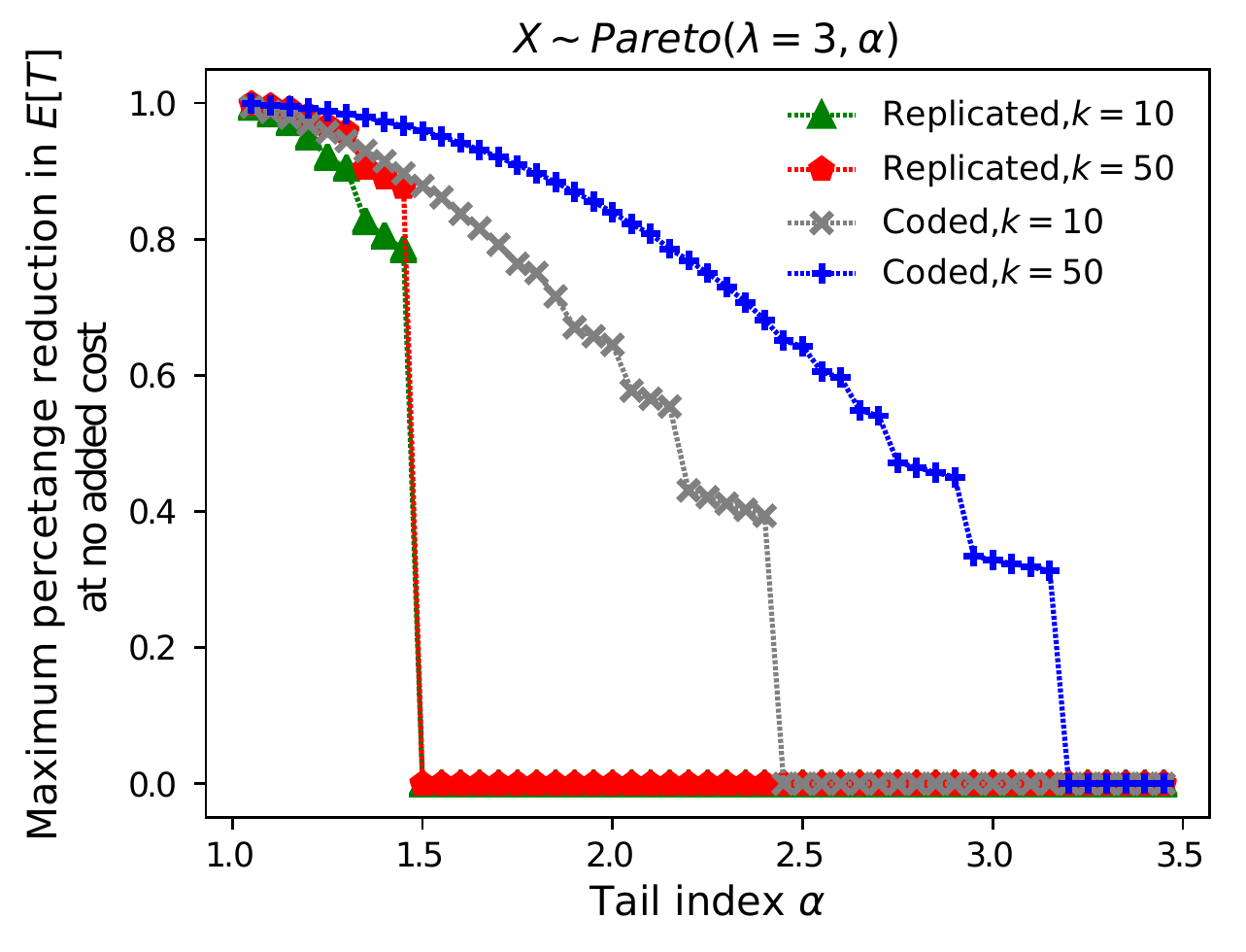}
  \caption{Maximum percentage reduction in expected latency without exceeding the baseline cost of running with no redundancy depends on the tail of task execution time. Latency reduction at no cost is possible for replicated redundancy only if the tail index is below $1.5$ while this constraint is relaxed under coding.}
  \label{fig:plot_zerodelay_red_pareto_reduc_in_E_T_for_baseline_E_C}
\end{figure}

\vspace{1ex}
\noindent\textbf{Demonstration using Google traces:} We simulated expected cost and latency of redundancy systems by using the empirical task execution time distributions that we constructed using Google Trace data \cite{reiss2011google}. Fig.~\ref{fig:plot_reped_vs_coded_Google} plots cost and latency curves using the three empirical distributions plotted in Fig.~\ref{fig:plot_google_empiricaltail} for jobs with $k=15, 400, 1050$ tasks.

For all three distributions, coding is doing better than replication in cost vs. latency tradeoff. Cost and latency can be reduced together with redundancy for all distributions because they pronounce a tail heavier than Exponential at large values. Note that although replication cannot reduce cost below the baseline cost of running with no redundancy, coding can achieve this for $k=15, 1050$. Also for $k=400$, although replication seems to achieve less cost and latency for low redundancy levels, coding outperforms replication beyond a certain level of redundancy.
\begin{figure*}[t]
  \centering
  \begin{subfigure}[]{.32\textwidth}
    \centering
    \includegraphics[width=1\textwidth, keepaspectratio=true]{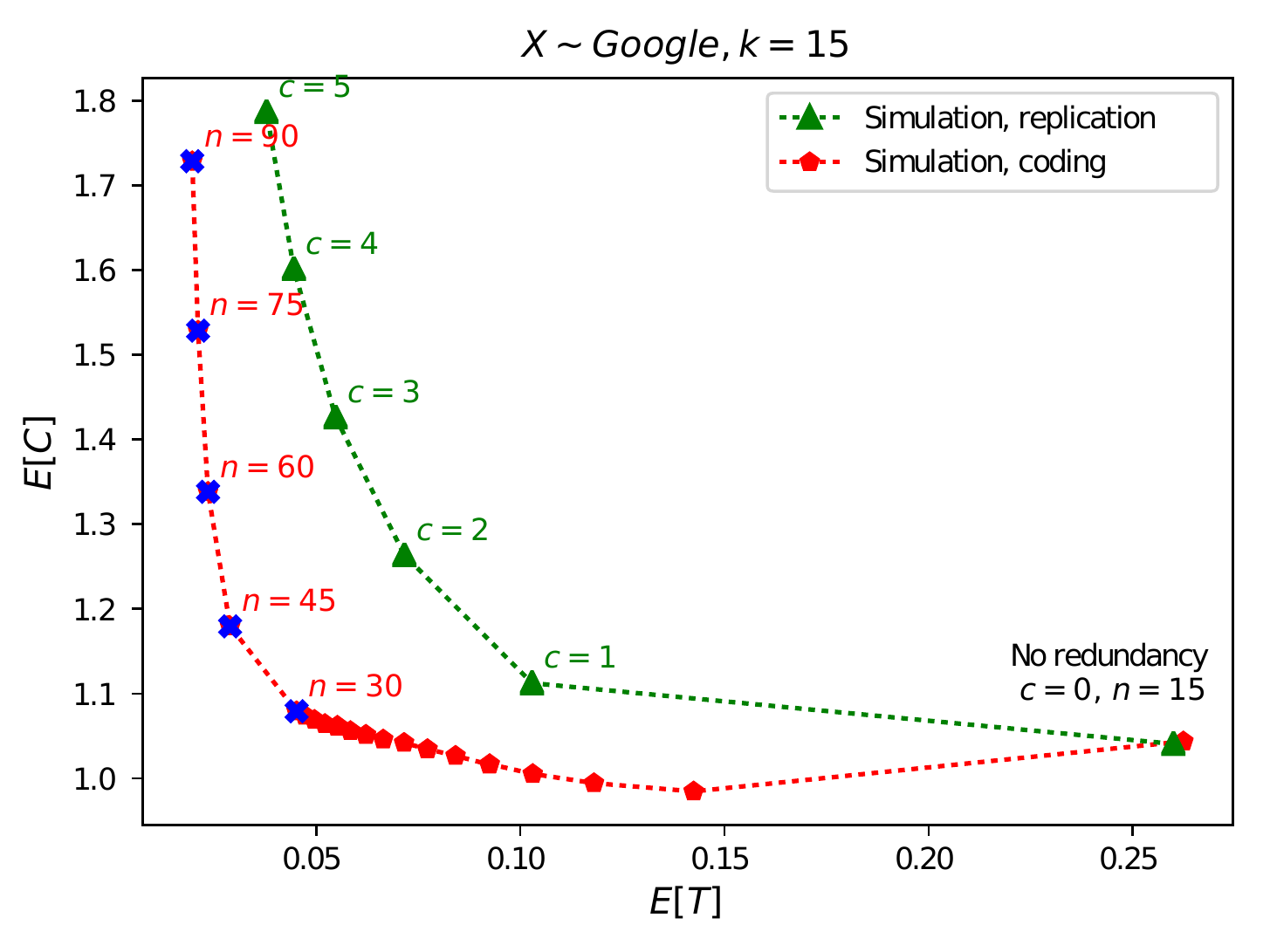}
  \end{subfigure}
  \begin{subfigure}[]{.32\textwidth}
    \centering
    \includegraphics[width=1\textwidth, keepaspectratio=true]{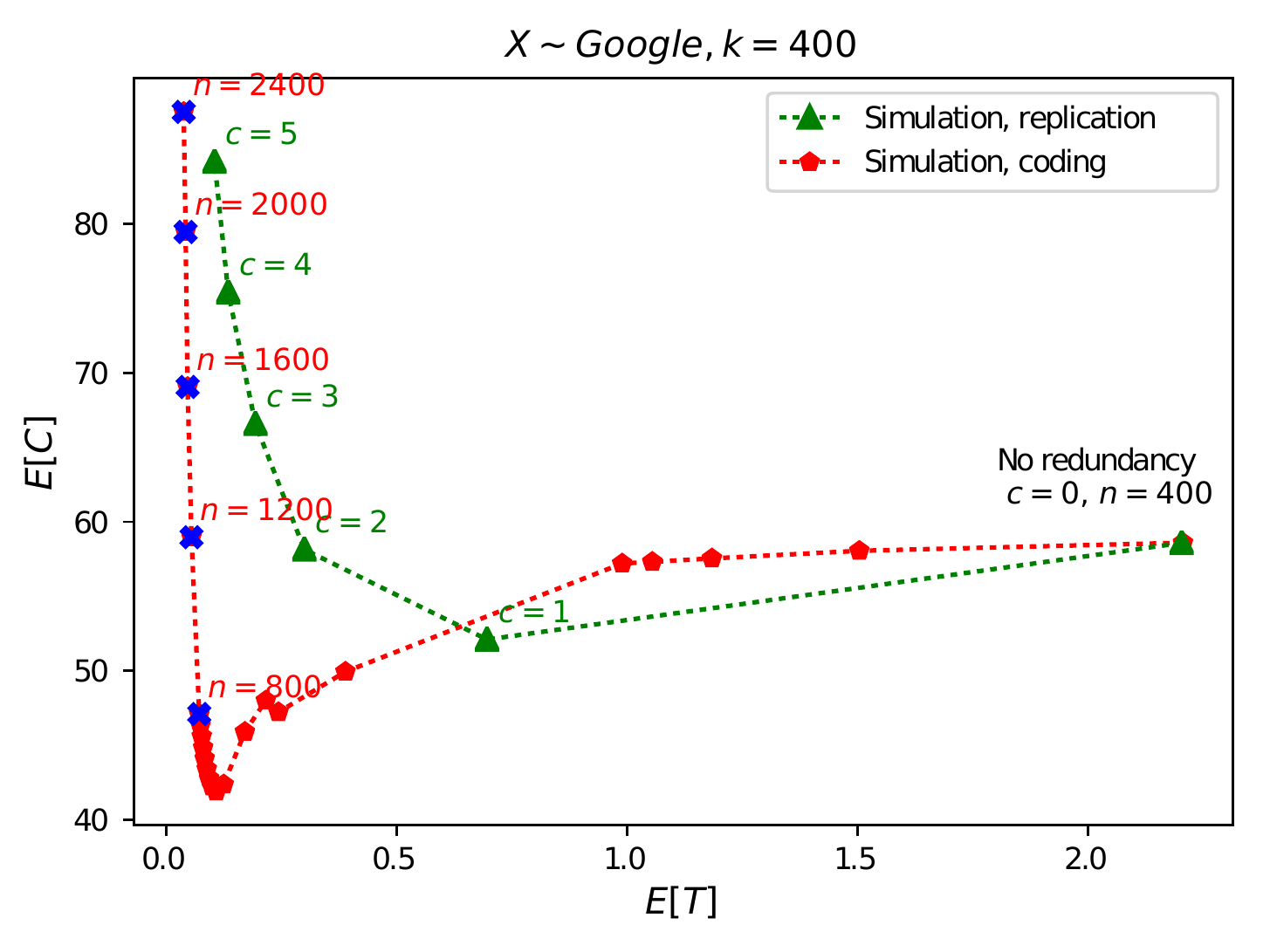}
  \end{subfigure}
  \begin{subfigure}[]{.32\textwidth}
    \centering
    \includegraphics[width=1\textwidth, keepaspectratio=true]{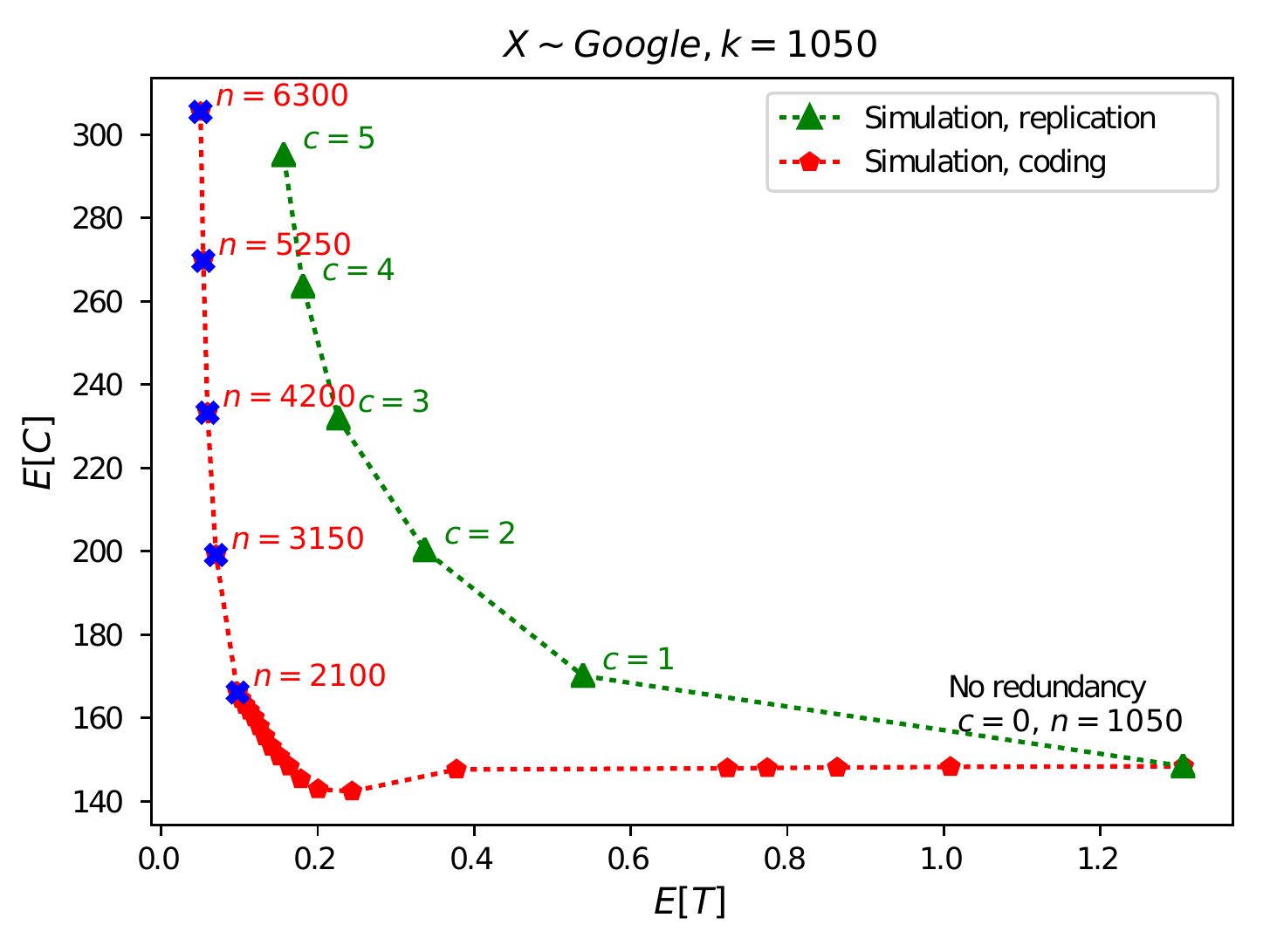}
  \end{subfigure}
  \caption{Expected cost vs.\ latency in zero-delay redundancy systems using empirical task execution time distribution for jobs with number of tasks $k=15, 400, 1050$ from Google Trace data.}
  \label{fig:plot_reped_vs_coded_Google}
\end{figure*}

\section{Straggler Relaunch}
There are two properties of heavy tailed distributions (e.g., {\it Pareto}) that greatly affect the latency of distributed job computation \cite{crovella2001performance}. First, if task execution times are heavy tailed, the longer a task has taken to execute, the longer its residual lifetime is expected to be. Second, majority of the mass in a set of sample observations drawn from a heavy tailed distribution is contributed by only a few samples. This suggests that among many heavy tailed tasks, very few of them are expected to be stragglers with very long completion time compared to the non-stragglers.

When task execution times are heavy tailed, after non-straggler tasks finish, we are expected to wait even longer for job completion since the remaining straggler tasks are expected to take at least as much more time as the non-stragglers have. This suggests that killing straggler tasks and launching fresh replacements can reduce the job completion time.

In this section, we study relaunching remaining tasks after some delay, and show that it can yield significant reduction in cost and latency when the tail of task execution time is heavy enough. We discuss about the level of tail heaviness required for relaunching to be useful. In the system we study, selection of tasks to relaunch is decided by the time relaunching is performed. Untimely relaunch may cause reduction in gain or even pain by either late relaunch and delayed cancellation of stragglers or early relaunch and killing non-stragglers. We present an approximate optimal relaunch time given the distribution of task execution times. Lastly, cost and latency of adding redundancy together with task relaunch is discussed.

\vspace{1ex}
\noindent
\textbf{No redundancy with relaunch:}
Thm.~\ref{thm_k_wrelaunch_E_T} gives exact expressions for expected latency and cost in no-redundancy systems, in which tasks that did not complete by time $\Delta$ are relaunched without introducing any redundancy. Note that we assume that relaunching takes place instantly upon cancellation and thus, it does not incur additional delay or cost. Fig.~\ref{fig:fig_k_wrelaunch_E_T__E_C_a} compares the latency with and without relaunch in a no-redundancy system. Relaunching tasks before the minimum task completion time $\lambda$ causes work loss and increases latency, while relaunching at the right time gives significant reduction in job completion time. Since cost is a direct function of latency in the absence of redundancy, reduction in latency will certainly reduce cost, as plotted in Fig.~\ref{fig:fig_k_wrelaunch_E_T__E_C_b}. Notice that relaunching all tasks at the beginning ($\Delta=0$) or not relaunching at all ($\Delta \rightarrow \infty$) implements the identical behavior and gives the identical cost and latency.
\begin{theorem}
  Under task execution time $X \sim Pareto(\lambda, \alpha)$ in a no-redundancy system (i.e., $c=0$ or $n=k$) with relaunch, the tail probability  of a job completion time is
  \begin{align}
    Pr\{T > t\} &= 1 - \Bigl\{\mathbbm{1}(t > \lambda)\Bigl[1 - \Bigl(\frac{\lambda}{t}\Bigl)^{\alpha}\Bigr]\Bigr\}^k
    \nonumber \\
    &+ \Bigl\{q + \mathbbm{1}(t > \Delta)\Bigl[1 - \Bigl(\frac{\Delta}{t}\Bigr)^{\alpha}\Bigr](1-q)\Bigr\}^k \label{eq:eq_k_wrelaunch_Pr_T_g_t}\\
    &+ \Bigl\{q + \mathbbm{1}(t > \Delta+\lambda)\Bigl[1 -\Bigl (\frac{\lambda}{t-\Delta}\Bigl)^{\alpha}\Bigr](1-q)\Bigr\}^k.
    \nonumber
  \end{align}
  
  The expected job completion time is
  \begin{equation}
  \begin{split}
    E[T] =
    \begin{cases}
      \Delta + g(k,\alpha) & \Delta \leq \lambda, \\
      \begin{split}
        & \Delta(1-q^k) + \\
        & g(k,\alpha)\Bigr[\Bigl(\frac{\lambda}{\Delta}-1)I(1-q;1-\alpha^{-1},k\Bigr) + 1\Bigr]
      \end{split} & o.w.
    \end{cases}
  \end{split}
  \label{eq:eq_k_wrelaunch_E_T}
  \end{equation}
  
  The expected cost with ($E[C^c]$) and without ($E[C]$) task cancellation is
  \begin{equation}
  \begin{split}
    E[C] &=
    \begin{cases}
      k\Delta + \frac{k\lambda}{1-\alpha^{-1}} & \Delta \leq \lambda, \\
      \frac{\alpha}{\alpha-1}[k\lambda(2-q)] - \frac{k\Delta(1-q)}{\alpha-1} & o.w.
    \end{cases} \\
    E[C^c] &=
    \begin{cases}
      k\Delta + \frac{1}{\alpha-1}[k\lambda\alpha - g(k,\alpha)] & \Delta \leq \lambda, \\
      \frac{\alpha}{\alpha-1}[k(1-q)(\lambda-\Delta) + k\lambda]+ k(1-q)\Delta & o.w.
    \end{cases}
  \end{split}
  \end{equation}
  where $q = \mathbbm{1}(\Delta > \lambda)(1 - (\frac{\lambda}{\Delta})^{\alpha})$ and $g(k,\alpha) = \lambda k!\frac{\Gamma(1-\alpha^{-1})}{\Gamma(k+1-\alpha^{-1})}$, which is the expected job completion time without relaunch.
  \label{thm_k_wrelaunch_E_T}
\end{theorem}
\begin{proof}[Sketch]
  Defining random variable $R$ as the number of tasks completed before $\Delta$, $R \sim Binomial(k, q)$ where $q = \mathbbm{1}(\Delta > \lambda)(1 - (\frac{\lambda}{\Delta})^{\alpha})$. Derivation of the tail, and expected latency and cost follows from the law of total probability or expectation by conditioning on $R$.
\end{proof}

An approximation for the optimal relaunch delay $\Delta^*$ that achieves the lowest latency and cost is given in Corollary \ref{cor_k_wrelaunch_E_T__opt_d_suff_a}. Convergence of this approximation to the true optimal is at a rate exponential with increasing $k$. Approximate $\Delta^*$ is very close to true optimal for $k=100$ as shown in Fig.~\ref{fig:fig_k_wrelaunch_E_T__E_C}. The optimal relaunch delay is an increasing function of minimum task execution time $\lambda$ and number of tasks, which intuitively makes sense. Also, it is a decreasing function of $\alpha$, meaning that it is better to relaunch earlier when the tail of task execution times is lighter, while for heavy tail, delaying relaunch further helps to identify stragglers and performs better in terms of cost and latency. This is because relaunching is a choice of canceling work that is already completed to get possibly lucky and execute fresh copies much faster than the canceled stragglers. Expected gain from relaunching under light tail is less than that under heavier tail since heavier tailed stragglers are expected to take much longer. Therefore, when the tail is light, it is better to try our chance with relaunch earlier and decrease amount of work loss with task cancellation.
\begin{corollary}
  Under task execution time $ Pareto(\lambda, \alpha)$ in a no-redundancy system, a sufficient condition on $\alpha$, which guarantees that expected cost and latency can be reduced by relaunching tasks at some time $\Delta$ is
  \begin{equation}
    \alpha < \frac{\ln(k+1)}{\ln(4)}.
  \label{eq:eq_k_wrelaunch_suff_a}
  \end{equation}
  Optimal relaunch time to achieve minimum cost and latency is approximated as
  \begin{equation}
    \Delta^* \approx \lambda\sqrt{\frac{k!\Gamma(1-\alpha^{-1})}{\Gamma(k+1-\alpha^{-1})}}
  \label{eq:eq_k_wrelaunch_optimal_d}
  \end{equation}
  which implies that optimal fraction of tasks to relaunch on average can be approximated as
  \begin{equation}
    p^* \approx \frac{\Gamma(1-\alpha^{-1})^{-\alpha/2}}{\sqrt{k+1}}.
  \label{eq:eq_k_wrelaunch_optimal_fraction}
  \end{equation}
  Upper bound on $\alpha$, and approximations for $\Delta^*$ and $p^*$ get tighter as the number of tasks $k$ increases.
  \label{cor_k_wrelaunch_E_T__opt_d_suff_a}
\end{corollary}
\begin{proof}
  Approximation for $\Delta^*$ follows from approximation $I(1-q, 1-\alpha^{-1}, k) \approx 1$ for large $k$ and then minimizing expected latency given in eq.~\eqref{eq:eq_k_wrelaunch_E_T} by taking derivative with respect to $\Delta$.
  
  Approximation for $p^*$ follows by observing that number of tasks completed before $\Delta^*$ is $R \sim Binomial(k, q^*)$ where $q^* = 1 - (\frac{\lambda}{\Delta^*})^{\alpha}$. Then, average fraction of tasks that are relaunched is $p^* = 1 - q^* = (\frac{\lambda}{\Delta^*})^{\alpha}$ in which we can use approximation $\Delta^* \approx \lambda\sqrt{\Gamma(1 - \alpha^{-1})(k+1)^{\alpha^{-1}}}$ for large $k$.
  
  Secondly, we show the sufficient condition on $\alpha$ to be able to reduce $E[T]$ with relaunch. Let $T_{norel}$ be job completion time in no-redundancy system with no relaunch, then $E[T_{norel}] = g(k,\alpha)$ where $g(k,\alpha) = \lambda\frac{k!\Gamma(1-\alpha^{-1})}{\Gamma(k+1-\alpha^{-1})}$.
  \begin{equation}
  \begin{split}
    E[T - T_{norel}] &= \Delta(1-q^k) \\
    &\quad - (1-\frac{\lambda}{\Delta})g(k,\alpha)I(1-q;1-\alpha^{-1},k) \\
    &\approx \Delta - (1-\frac{\lambda}{\Delta})g(k,\alpha).
  \end{split}
  \end{equation}
  This difference is smallest when $\Delta = \Delta^*$ for which we will use the approximate discussed above. We are interested in the maximum value of $\alpha$ that would allow maximum possible difference $E[T - T_{norel}]$ to be negative as
  \begin{equation*}
  \begin{split}
    \Delta^* - (1-\frac{\lambda}{\Delta^*})g(k,\alpha) &< 0 \\
    2\sqrt{\lambda g(k,\alpha)} - g(k,\alpha) &< 0 \\
    \frac{\Gamma(k+1)\Gamma(1-\alpha^{-1})}{\Gamma(k+1-\alpha^{-1})} &> 4 \\
    \ln(\frac{\Gamma(k+1)}{\Gamma(k+1-\alpha^{-1})}) + \ln(\Gamma(1-\alpha^{-1})) &> 4 \\
    \frac{\ln(k+1)}{\alpha} + \ln(\Gamma(1-\alpha^{-1})) &\stackrel{(a)}{>} \ln(4).
  \end{split}
  \end{equation*}
  where $(a)$ is by using the approximation $\frac{\Gamma(k+1)}{\Gamma(k+1-\alpha^{-1})} \approx (k+1)^{1/\alpha}$ for large $k$. For $\alpha \geq 1$, which is what we assume since it is a requirement for finite expected latency, $\ln(\Gamma(1-\alpha^{-1})) > 0$ holds, and so sufficient condition given in eq.~\eqref{eq:eq_k_wrelaunch_suff_a} follows.
\end{proof}

Expression for optimal delay $\Delta^*$ tells us that it is better to relaunch earlier when task execution times have lighter tail. However, relaunching earlier does not mean that more tasks will be relaunched. Fraction of tasks $p^*$ that are relaunched after optimal delay $\Delta^*$ monotonically decreases with $\alpha$~\footnote{$p^*$ is a monotonically decreasing function of $\alpha$. For very heavy tail i.e., $\lim_{\alpha \rightarrow 1} \Gamma(1-\alpha^{-1})^{-\alpha/2} = 1$, for very light tail i.e., $\lim_{\alpha \rightarrow \infty} \Gamma(1-\alpha^{-1})^{-\alpha/2} \approx 0.749$.}, i.e., as the tail gets lighter. Notice that $p^*$ decreases with $k$, which means for jobs with larger number of tasks, optimal strategy dictates relaunching smaller fraction of the tasks. For instance, suppose $\alpha=2$ and $k=100$, then $p^* \approx 0.06$, in other words, only $6\%$ of the tasks would need to be relaunched on average for optimal latency and cost.

Note that we assume relaunching takes place instantly and does not introduce any cost. Adding relaunch cost into the analysis, which we leave as a future work, could make the analysis more realistic and also give more insight in searching for optimal relaunch strategy in practice.

\captionsetup[subfigure]{labelformat=empty}
\begin{figure}[h]
  \centering
  \begin{subfigure}[]{.35\textwidth}
    \centering
    \includegraphics[width=1\textwidth, keepaspectratio=true]{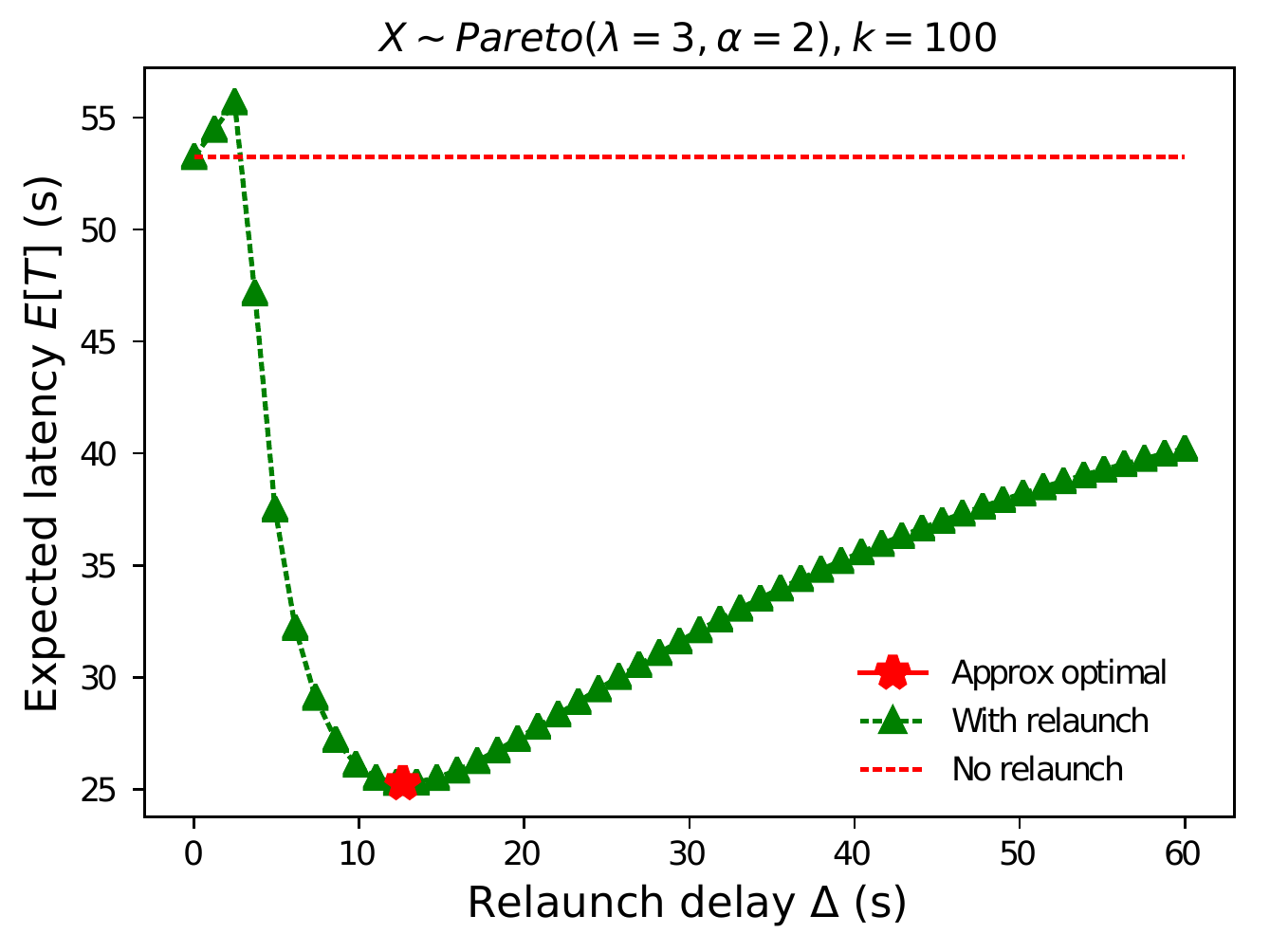}
    \caption{}
    \label{fig:fig_k_wrelaunch_E_T__E_C_a}
  \end{subfigure}
  \begin{subfigure}[]{.35\textwidth}
    \centering
    \includegraphics[width=1\textwidth, keepaspectratio=true]{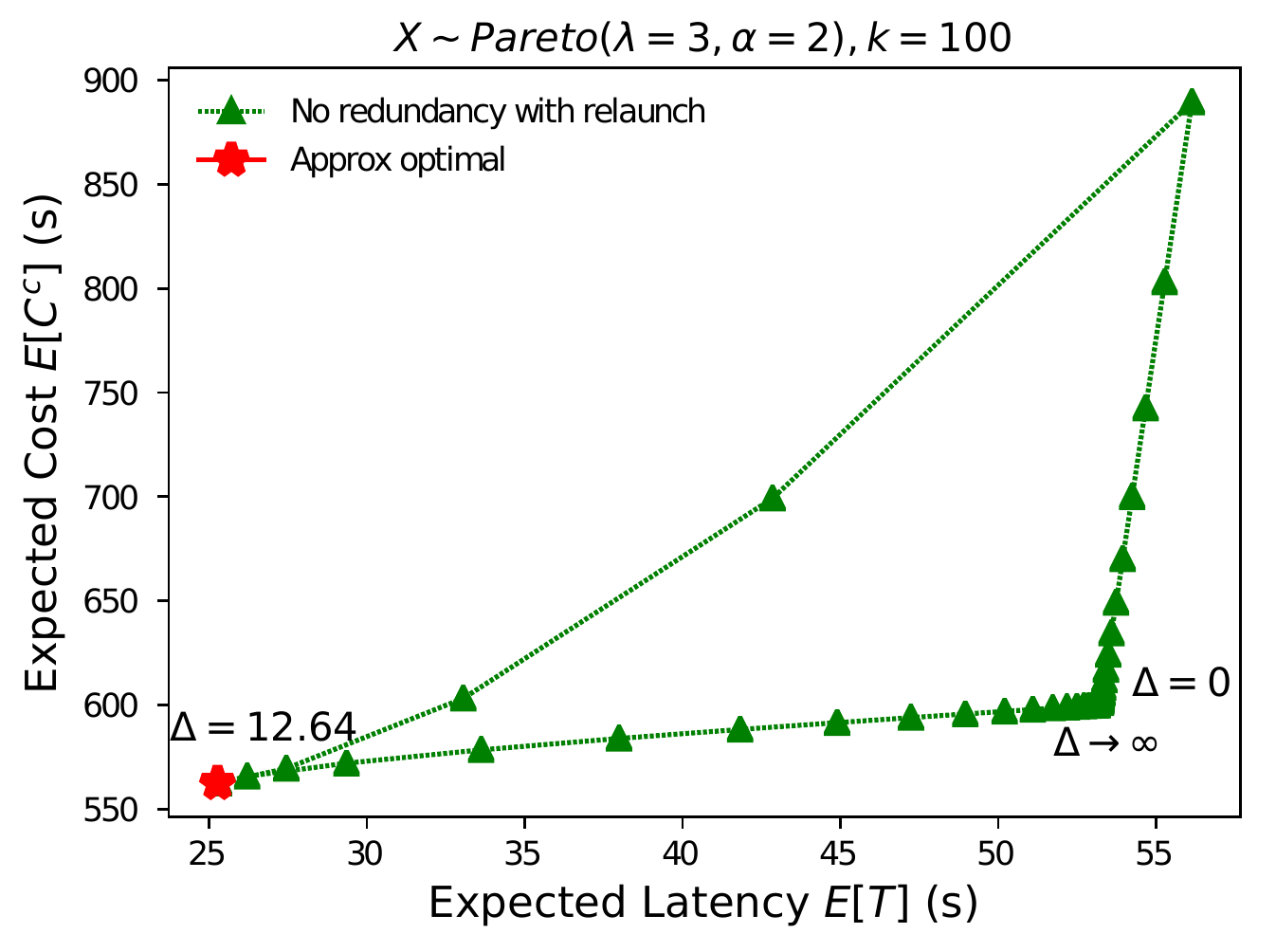}
    \caption{}
    \label{fig:fig_k_wrelaunch_E_T__E_C_b}
  \end{subfigure}
  \caption{(Top) Expected latency in no-redundancy system with and without relaunch in terms of relaunch delay $\Delta$. (Bottom) Expected cost vs. latency curve for no-redundancy system with relaunch, along which the value of delay $\Delta$ is varied.}
  \label{fig:fig_k_wrelaunch_E_T__E_C}
\end{figure}


\begin{figure}[b]
  \centering
  \includegraphics[width=0.35\textwidth, keepaspectratio=true]{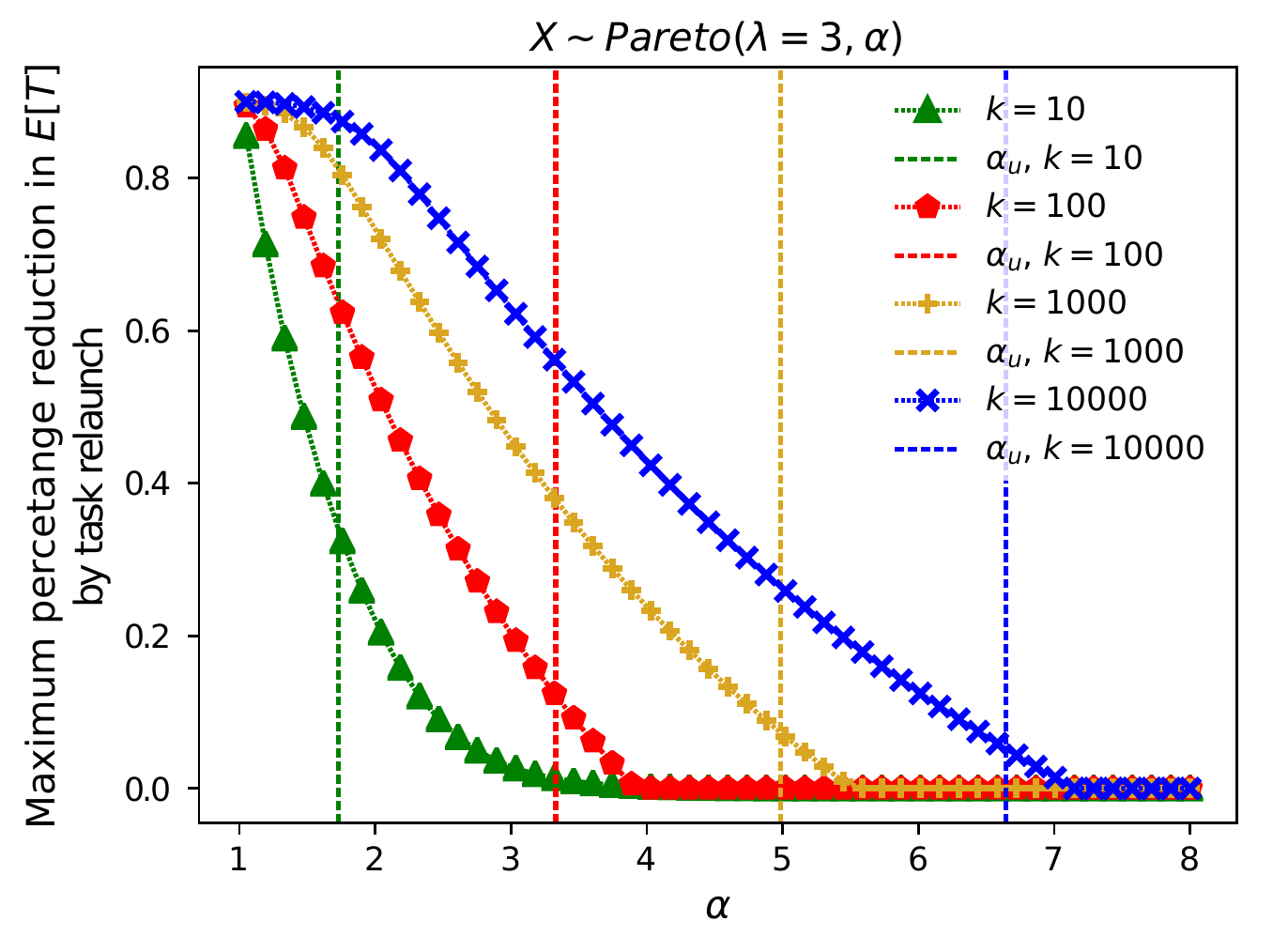}
  \caption{Maximum percentage reduction in expected latency by task relaunch depends on the tail of the task execution time. $\alpha_u$ is the upper bound on $\alpha$ given as the sufficient condition in Corollary \ref{cor_k_wrelaunch_E_T__opt_d_suff_a}.}
  \label{fig:fig_k_wrelaunch__reduc_in_E_T_vs_a}
\end{figure}

\begin{figure*}[t]
  \centering
  \begin{subfigure}[]{.32\textwidth}
    \centering
    \includegraphics[width=1\textwidth, keepaspectratio=true]{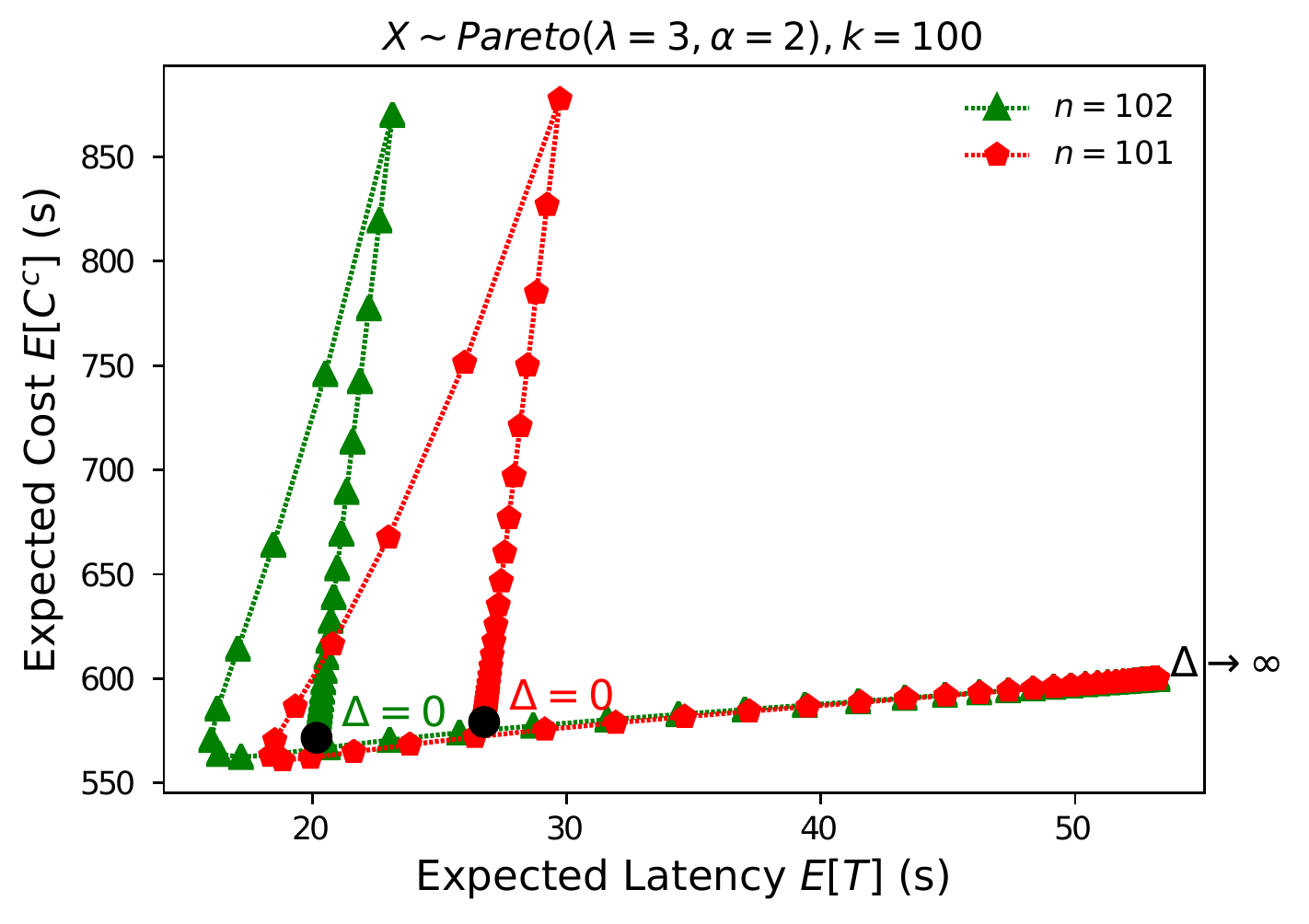}
  \end{subfigure}
  \begin{subfigure}[]{.32\textwidth}
    \centering
    \includegraphics[width=1\textwidth, keepaspectratio=true]{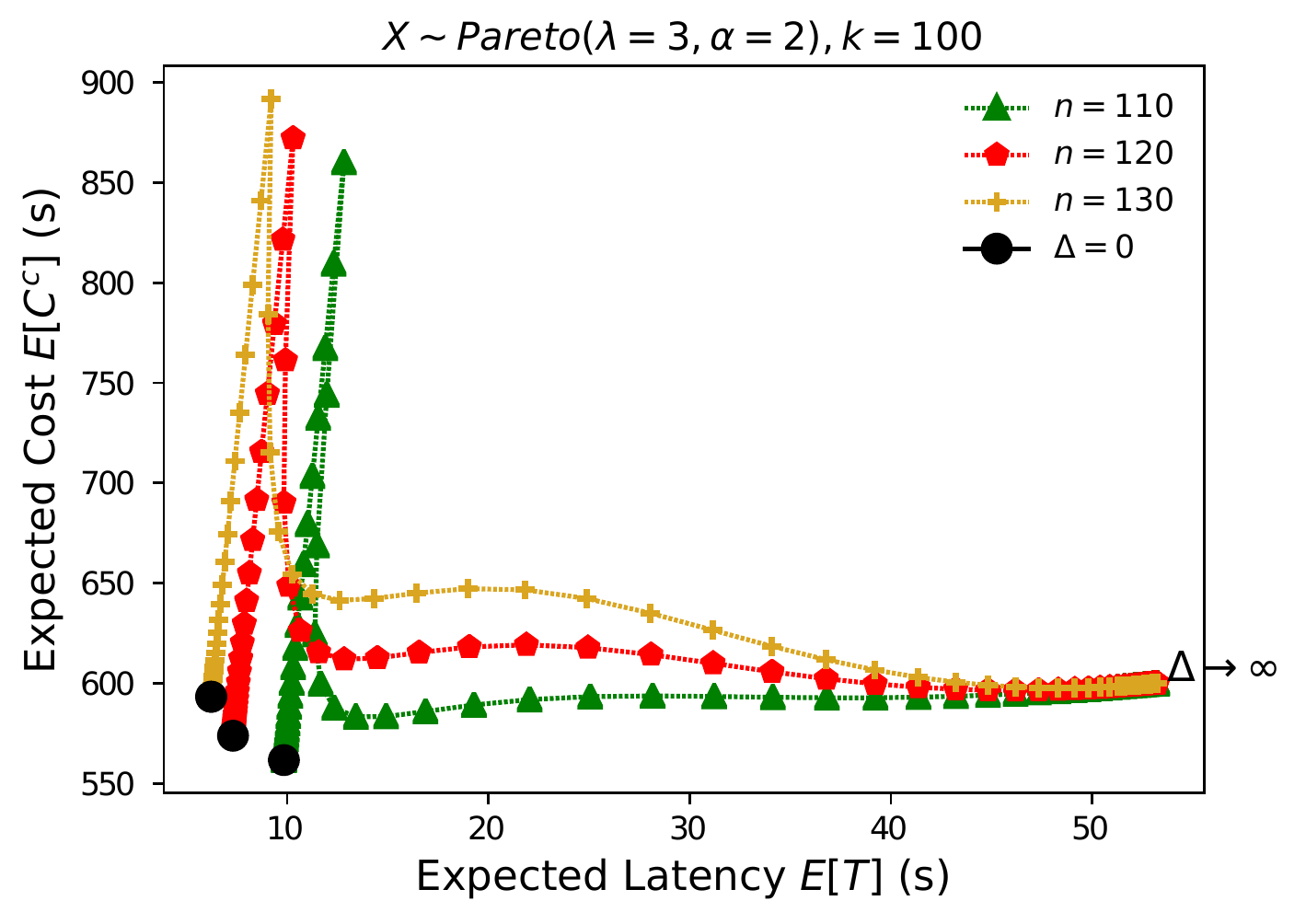}
  \end{subfigure}
  \begin{subfigure}[]{.32\textwidth}
    \centering
    \includegraphics[width=1\textwidth, keepaspectratio=true]{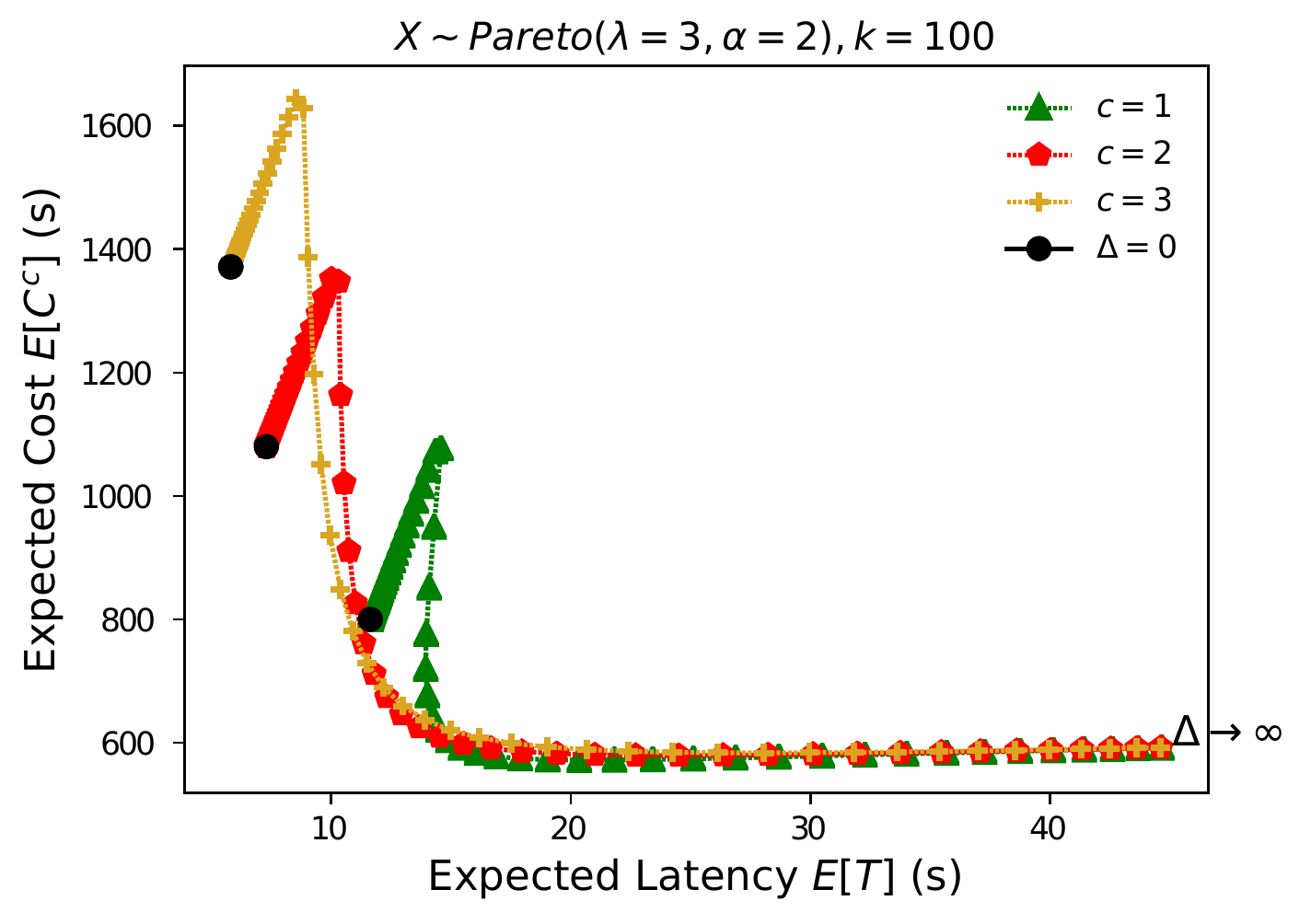}
  \end{subfigure}
  \caption{Expected cost vs.\ latency for coding $(k, n, \Delta)$ and replication $(k, c, \Delta)$ systems with relaunch. Each curve is plotted by interpolating between incremental steps of $\Delta$. Delaying redundancy is effective to reduce cost under low level of coded redundancy (Left) while it is not effective under higher level of coded redundancy (Middle) and not effective in replicated redundancy at all times (Right).}
  \label{fig:fig_k_nc_wrelaunch_E_C_vs_E_T}
\end{figure*}

Relaunching allows reducing the average number of stragglers by replacing tasks that appear to be straggling with fresh copies. For relaunching to be effective, loss incurred by starting fresh copies from scratch should be compensated by avoiding very long execution times of stragglers. In other words, for relaunching to be able to reduce latency and cost, task execution times must be heavy tailed beyond a threshold. If the tail is lighter than this threshold, relaunching actually hurts and increases cost and latency (e.g., relaunching always hurts when task execution times are light tailed). Corollary \ref{cor_k_wrelaunch_E_T__opt_d_suff_a} gives a sufficient condition on tail index $\alpha$ such that for any $\alpha$ less than $\frac{\ln(k+1)}{\ln(4)}$, relaunching helps to reduce cost and latency. Note that this upper bound does not depend on the minimum task completion time $\lambda$ and is only proportional to the logarithm of the number of tasks $k$, which we also validated by numerically computing the exact upper limit on $\alpha$. This upper bound given as the sufficient condition and the exact upper limit on $\alpha$ get closer as $k$ increases, which is illustrated in Fig.~\ref{fig:fig_k_wrelaunch__reduc_in_E_T_vs_a}.

\vspace{1ex}
\noindent
\textbf{Redundancy with relaunch:} Here we study the effect of adding redundancy together with relaunching remaining tasks at time $\Delta$. We modify the previously studied redundancy systems as the follows. In replication $(k, c, \Delta)$-system, each remaining task at time $\Delta$ is relaunched together with $c$ new replicas. In coding $(k, n, \Delta)$-system, each remaining task at time $\Delta$ is relaunched and overall $n-k$ new parity tasks are added.

Thm.~\ref{thm_k_cn_wrelaunch_E_T__E_C} state exact expressions for the expected cost and latency in respectively replicated and coded redundancy systems with relaunch. Fig.~\ref{fig:fig_k_nc_wrelaunch_E_C_vs_E_T} plots cost vs.\ latency curves for varying level of redundancy. When level of coded redundancy is low (e.g., $n-k = 1,2$), there is an optimum delay $\Delta$ that gives the minimum cost and latency as observed previously for no-redundancy system with relaunch. As the level of coded redundancy increases (e.g., $n-k \geq 10$), redundancy becomes a greater effect on cost and latency than relaunching stragglers, and delaying redundancy becomes not effective in reducing cost as observed previously for redundancy systems with no relaunch (see \cite{ourmamapaper}). In replication system, delaying redundancy is ineffective to reduce cost at all times. This is because replicating each remaining task even by one is enough to dominate relaunching stragglers in terms of the effect on cost and latency.

\begin{theorem}
  Suppose task execution time is $Pareto(\lambda, \alpha)$. In replication $(k, c, \Delta)$-system with relaunch, expected job execution time is
  \begin{equation*}
  \begin{split}
    E[T] =
    \begin{cases}
      \Delta + \lambda k!\frac{\Gamma(1-\tilde{\alpha}^{-1})}{\Gamma(k+1-\tilde{\alpha}^{-1})} & \Delta \leq \lambda, \\
      \begin{split}
        & \lambda\frac{\Gamma(1-\tilde{\alpha}^{-1})}{\Gamma(-\tilde{\alpha}^{-1})}B(k-kq+1, -\tilde{\alpha}^{-1}) \\
        & - \lambda\frac{\Gamma(1-\alpha^{-1})}{\Gamma(-\alpha^{-1})}B(k-kq+1, -\alpha^{-1}) \\
        &+ E[T_{no}].
      \end{split} & o.w.
    \end{cases}
  \end{split}
  \label{eq:eq_k_c_wrelaunch_E_T}
  \end{equation*}
  where $E[T_{no}]$ is the expected job completion time for no-redundancy system with relaunch as given in eq.~\eqref{eq:eq_k_wrelaunch_E_T}.
  
  Expected cost with ($C^c$) and without ($C$) task cancellation is
  \begin{equation*}
  \begin{split}
    E[C^c] &=
    \begin{cases}
      k\Delta + k\lambda(c+1)\frac{\tilde{\alpha}}{\tilde{\alpha}-1} & \Delta \leq \lambda, \\
      \begin{split}
        & \frac{k\alpha}{(\alpha-1)}(\lambda - \Delta(1-q)) \\
        &+ k(1-q)\Delta + k\lambda(c+1)(1-q)\frac{\tilde{\alpha}}{\tilde{\alpha}-1}
      \end{split}
      & o.w.
    \end{cases} \\
    E[C] &=
    \begin{cases}
      k\Delta + k\lambda(c+1)\frac{\alpha}{\alpha-1} & \Delta \leq \lambda, \\
      \begin{split}
        & \frac{k\alpha}{(\alpha-1)}(\lambda - \Delta(1-q)) + k(1-q)\Delta \\
        &+ k\lambda(c+1)(1-q)\frac{\alpha}{\alpha-1}
      \end{split}
      & o.w.
    \end{cases}
  \end{split}
  \label{eq:eq_k_c_wrelaunch_E_C}
  \end{equation*}
  where $\tilde{\alpha} = (c+1)\alpha$ and $q = \mathbbm{1}(\Delta > \lambda)(1 - (\frac{\lambda}{\Delta})^{\alpha})$.
  In coding $(k, n, \Delta)$-system with relaunch, expected job execution time is
  \begin{equation*}
  \begin{split}
    E[T] =
    \begin{cases}
      \Delta + \lambda \frac{n!}{(n-k)!}\frac{\Gamma(n-k+1-\alpha^{-1})}{\Gamma(n+1-\alpha^{-1})} & \Delta \leq \lambda, \\
      \begin{split}
        & \Delta(1-q^k) + \lambda(\frac{B(n-kq+1,-\alpha^{-1})}{B(n-k+1,-\alpha^{-1})} \\
        &+ kB(q;k,1-\alpha^{-1}) - q^k)
      \end{split} & o.w.
    \end{cases}
  \end{split}
  \label{eq:eq_k_n_wrelaunch_E_T}
  \end{equation*}
  
  Expected cost with ($C^c$) and without ($C$) task cancellation is
  \begin{equation}
  \begin{split}
    E[C^c] &=
    \begin{cases}
      k\Delta + \lambda\frac{n}{\alpha-1}(\alpha - \frac{\Gamma(n)}{\Gamma(n-k)}\frac{\Gamma(n-k+1-\alpha^{-1})}{\Gamma(n+1-\alpha^{-1})}) & \Delta \leq \lambda, \\
      \begin{split}
        & \frac{\alpha}{\alpha-1}(k(1-q)(\lambda-\Delta) + n\lambda) \\
        &+ k(1-q)\Delta - \lambda(n-k)q^k \\
        &- \frac{\lambda}{\alpha-1}(n-k)\frac{B(n-kq+1, -\alpha^{-1})}{B(n-k+1, -\alpha^{-1})}.
      \end{split} & o.w.
    \end{cases} \\
    E[C] &=
    \begin{cases}
      k\Delta + \frac{n\lambda}{1-\alpha^{-1}} & \Delta \leq \lambda, \\
      \begin{split}
      & \frac{\alpha}{\alpha-1}(k\lambda(1-q+q^k) + n\lambda(1-q^k)) \\
      &- \frac{k\Delta(1-q)}{\alpha-1}
      \end{split} & o.w.
    \end{cases}
  \end{split}
  \label{eq:eq_k_n_wrelaunch_E_Cnocancel}
  \end{equation}
  where $q = \mathbbm{1}(\Delta > \lambda)(1 - (\frac{\lambda}{\Delta})^{\alpha})$.
  
\label{thm_k_cn_wrelaunch_E_T__E_C}
\end{theorem}

\section{Open Problems}
In the analysis given here, we assumed that execution time of each task is iid  and does not depend on resources they run on. Main justification for this assumption is that computing systems usually have large number of nodes and each task can be placed on a separate node at random. In other words, we assumed that the task execution time is decoupled from resource scheduling. However, execution time of a task that exclusively runs on a node would have different distribution than a task that runs together with several others on the same node. There are two challenges in the analysis of latency and cost with resource scheduling in mind: 1) There is not enough experimental evidence to accurately model execution time with respect to resource load, 2) Order statistics of multiple distribution families are intractable.

Moreover, in reality, task execution times cannot be arbitrarily large, but modeling as {\it SExp} or {\it Pareto} does not implement this restriction. Another possible extension of the cost and latency analysis would be considering right truncated execution time models. Main challenge for this extension is that truncation renders the order statistics analysis tedious and often intractable.



\begin{thebibliography}{10}

\bibitem{al2013cloud}
{\sc Al-Roomi, M., Al-Ebrahim, S., Buqrais, S., and Ahmad, I.}
\newblock Cloud computing pricing models: a survey.
\newblock {\em International Journal of Grid and Distributed Computing 6}, 5
  (2013), 93--106.

\bibitem{ananthanarayanan2013effective}
{\sc Ananthanarayanan, G., Ghodsi, A., Shenker, S., and Stoica, I.}
\newblock Effective straggler mitigation: Attack of the clones.
\newblock In {\em NSDI\/} (2013), vol.~13, pp.~185--198.

\bibitem{Pareto:Arnold15}
{\sc Arnold, B.~C.}
\newblock {\em Pareto distribution}.
\newblock Wiley Online Library, 2015.

\bibitem{OrderStat:Arnold08}
{\sc Arnold, B.~C., Balakrishnan, N., and Nagaraja, H.~N.}
\newblock {\em A first course in order statistics}.
\newblock SIAM, 2008.

\bibitem{crovella2001performance}
{\sc Crovella, M.~E.}
\newblock Performance evaluation with heavy tailed distributions.
\newblock In {\em Workshop on Job Scheduling Strategies for Parallel
  Processing\/} (2001), Springer, pp.~1--10.

\bibitem{dean2013tail}
{\sc Dean, J., and Barroso, L.~A.}
\newblock The tail at scale.
\newblock {\em Communications of the ACM 56}, 2 (2013), 74--80.

\bibitem{dean2008mapreduce}
{\sc Dean, J., and Ghemawat, S.}
\newblock Mapreduce: simplified data processing on large clusters.
\newblock {\em Communications of the ACM 51}, 1 (2008), 107--113.
\newpage
\bibitem{dutta2016short}
{\sc Dutta, S., Cadambe, V., and Grover, P.}
\newblock Short-dot: Computing large linear transforms distributedly using
  coded short dot products.
\newblock In {\em Advances In Neural Information Processing Systems\/} (2016),
  pp.~2092--2100.
\bibitem{CodedGradientDescent:HalbawiAS17}
{\sc Halbawi, W., Azizan-Ruhi, N., Salehi, F., and Hassibi, B.}
\newblock Improving distributed gradient descent using reed-solomon codes.
\newblock {\em arXiv preprint arXiv:1706.05436\/} (2017).

\bibitem{joshi2015queues}
{\sc Joshi, G., Soljanin, E., and Wornell, G.}
\newblock Queues with redundancy: Latency-cost analysis.
\newblock {\em ACM SIGMETRICS Performance Evaluation Review 43}, 2 (2015),
  54--56.

\bibitem{ourmamapaper}
{\sc Akta\c{s}, M., Peng, P., Soljanin, E.}
\newblock Effective straggler mitigation: Which clones should attack and when?
\newblock In {\em MAMA Workshop SIGMETRICS\/} (2017).

\bibitem{reiss2012towards}
{\sc Reiss, C., Tumanov, A., Ganger, G.~R., Katz, R.~H., and Kozuch, M.~A.}
\newblock Towards understanding heterogeneous clouds at scale: Google trace
  analysis.
\newblock {\em Intel Science and Technology Center for Cloud Computing, Tech.
  Rep\/} (2012), 84.

\bibitem{reiss2011google}
{\sc Reiss, C., Wilkes, J., and Hellerstein, J.~L.}
\newblock Google cluster-usage traces: format+ schema.
\newblock {\em Google Inc., White Paper\/} (2011), 1--14.

\bibitem{wang2015using}
{\sc Wang, D., Joshi, G., and Wornell, G.}
\newblock Using straggler replication to reduce latency in large-scale parallel
  computing.
\newblock {\em ACM SIGMETRICS Performance Evaluation Review 43}, 3 (2015),
  7--11.

\end{thebibliography}
\end{document}